\documentclass[11pt]{article}

\usepackage{fullpage}
\usepackage{appendix}
\usepackage{amssymb}
\usepackage{amsfonts}
\usepackage{amsmath}
\usepackage{amsthm}
\usepackage{latexsym}
\usepackage{epsfig}
\usepackage{makeidx}
\usepackage{graphicx}
\usepackage[numbers]{natbib}
\usepackage{hyperref}
\usepackage{epstopdf}
\usepackage{setspace}

\newtheorem{theorem}{Theorem}[section]
\newtheorem{lemma}[theorem]{Lemma}
\newtheorem{remark}{Remark}

\newtheorem{definition}[theorem]{Definition}
\newtheorem{corollary}[theorem]{Corollary}
\newtheorem{proposition}[theorem]{Proposition}

\newtheorem{conjecture}[theorem]{Conjecture}

\newtheorem{prop}[theorem]{Proposition}

\newtheorem*{proposition*}{Proposition}

\newcommand{\Boolean}{\{0,1\}^n \rightarrow \{0,1\}}

\newcommand{\nexp}{\mathsf{NEXP}}
\renewcommand{\ne}{\mathsf{NE}}
\newcommand{\cone}{\mathsf{coNE}}
\newcommand{\iocone}{\mathsf{i.o.coNE}}
\newcommand{\conexp}{\mathsf{coNEXP}}
\newcommand{\tc}{\mathsf{TC}^0}
\newcommand{\acc}{\mathsf{ACC}}
\newcommand{\ac}{\mathsf{AC}^0}
\newcommand{\nc}{\mathsf{NC}^1}

\newcommand{\gand}{\mathsf{AND}}
\newcommand{\gor}{\mathsf{OR}}
\newcommand{\gmod}{\mathsf{MOD}}
\newcommand{\gnot}{\mathsf{NOT}}
\newcommand{\gxor}{\mathsf{XOR}}
\newcommand{\gequiv}{\mathsf{EQUIV}}
\newcommand{\classc}{\mathcal{C}}
\newcommand{\poly}{\mathsf{poly}}
\newcommand{\ntime}{\mathsf{NTIME}}
\newcommand{\equivandc}{\mathsf{Equiv}$-$\mathsf{AND}$-$\classc}
\newcommand{\equivand}{\mathsf{Equiv}$-$\mathsf{AND}}
\newcommand{\equivorc}{\mathsf{Equiv}$-$\mathsf{OR}$-$\classc}
\newcommand{\ssat}{\mathsf{Succinct}$-$\mathsf{SAT}}
\newcommand{\ppoly}{\mathsf{P}/\mathsf{poly}}
\newcommand{\perm}{\mathsf{PERM}}
\newcommand{\pit}{\mathsf{PIT}}
\newcommand{\nsubexp}{\mathsf{NSUBEXP}}
\newcommand{\size}{\mathsf{SIZE}}
\newcommand{\asize}{\mathsf{ASIZE}}

\newcommand{\propp}{\mathcal{P}}
\newcommand{\proppyes}{\propp_{\mathsf{yes}}}
\newcommand{\proppno}{\propp_{\mathsf{no}}}

\newcommand{\dtime}{\mathsf{DTIME}}
\newcommand{\matime}{\mathsf{MATIME}}

\begin{document}

\title{Algorithms versus Circuit Lower Bounds~\\~\\}

\author{Igor C. Oliveira\footnote{Research supported by NSF grants CCF-1116702 and CCF-1115703.}~\\~\\{\small Department of Computer Science}~\\{\small Columbia University}~\\~\\{\tt oliveira@cs.columbia.edu}~\\~\\~\\}

\maketitle

\begin{abstract}
Different techniques have been used to prove several transference theorems of the form ``nontrivial algorithms for a circuit class $\classc$ yield circuit lower bounds against $\classc$''. In this survey we revisit many of these results. We discuss how circuit lower bounds can be obtained from derandomization, compression, learning, and satisfiability algorithms. We also cover the connection between circuit lower bounds and useful properties, a notion that turns out to be fundamental in the context of these transference theorems. Along the way, we obtain a few new results, simplify several proofs, and show connections involving different frameworks. We hope that our presentation will serve as a self-contained  introduction for those interested in pursuing research in this area.
\end{abstract}

\newpage 

\onehalfspacing

\tableofcontents

\singlespacing

\newpage

\section{Introduction}

This survey deals with two fundamental problems in theoretical computer science: the design of nontrivial algorithms for difficult computational tasks, and the search for unconditional proofs that some natural computational problems are inherently hard (more specifically, do not admit polynomial size circuits).

Perhaps surprisingly, these problems are deeply related. For instance, it follows from the work of Karp and Lipton \citep{DBLP:conf/stoc/KarpL80} (attributed to Meyer) that if $3$-SAT admits a polynomial time algorithm, then there are problems solved in exponential time that cannot be computed by polynomial size circuits. On the other hand, it is known that constructive proofs of circuit lower bounds lead to algorithms breaking exponentially hard pseudorandom generators that are conjectured to exist (Razborov and Rudich \citep{DBLP:journals/jcss/RazborovR97}). 

The last decade has produced several additional \emph{transference theorems}\footnote{In other words, these theorems show that one can transform an algorithmic result into a circuit lower bound, i.e., they allow us to transfer a result from one area to another.} of this form, under many different algorithmic frameworks. For instance, the existence of subexponential time learning algorithms for a class of functions $\classc$ leads to circuit lower bounds against $\classc$ (Fortnow and Klivans \citep{DBLP:journals/jcss/FortnowK09}). In a different domain, it is known that the design of subexponential time deterministic algorithms for problems with efficient randomized algorithms implies circuit lower bounds that have eluded researchers for decades (Kabanets and Impagliazzo \citep{DBLP:journals/cc/KabanetsI04}). More recently, it has been shown that new circuit lower bounds can be obtained from efficient compression algorithms (Chen et al. \citep{DBLP:journals/eccc/ChenKKSZ13}), not to mention the connection between satisfiability algorithms and circuit lower bounds (Williams \citep{DBLP:conf/stoc/Williams10}, \citep{DBLP:conf/coco/Williams11}, \citep{DBLP:conf/stoc/Williams13}). Several additional results have appeared in the literature (\citep{DBLP:journals/cc/KinneMS12}, \citep{DBLP:journals/toc/AaronsonM11}, \citep{DBLP:journals/eccc/AaronsonABHM10}, \citep{DBLP:journals/cc/AydinliogluGHK11}, \citep{DBLP:conf/icalp/HarkinsH11}, \citep{KKO}, among others). For a gentle introduction to some of these connections, see Santhanam \citep{DBLP:journals/eccc/Santhanam12}.

These results are interesting for several reasons. For instance, as far as we know, there may be an efficient compression scheme that works well for any string possessing some structure, or circuits of polynomial size can be learned in quasipolynomial time by a very complicated learning algorithm. Nevertheless, the transference theorems discussed before show that even if complicated tasks like these are actually easy, some natural computational problems are inherently hard (even for non-uniform algorithms).

How could one prove (unconditionally) that a natural computational problem is hard? How would a mathematical proof of such result look like? This is one of the most fascinating questions of contemporary mathematics, and it is related to deep problems about algorithms, combinatorics, and mathematical logic (cf.  Kraj\'{i}cek \citep{krajicek}, Immerman \citep{DBLP:books/daglib/0095988}, Cook and Nguyen \citep{cooknguyen}).

It turns out that the connection between algorithms and circuit lower bounds (``transference theorems'') can be used to prove \emph{new} circuit lower bounds that had resisted the use of more direct approaches for decades. Let $\classc$ be a class of circuits, such as $\ac, \tc, \nc$, etc. We say that a satisfiability algorithm for $\classc$ is \emph{nontrivial} if it runs in time time $2^n/s(n)$, for some function $s(n) \gg \poly(n)$. Building on work done by many researchers, Williams (\citep{DBLP:conf/coco/Williams11}, \citep{DBLP:conf/stoc/Williams10}) proved the following transference theorem: the existence of a nontrivial $\classc$-SAT algorithm implies $\nexp \nsubseteq \classc[\poly]$. In other words, faster satisfiability algorithms lead to languages computed in nondeterminstic exponential time that cannot be computed by polynomial size circuits from $\classc$. 

Most importantly, by designing a new $\acc$-SAT algorithm, Williams \citep{DBLP:conf/coco/Williams11} was able to obtain a circuit lower bound for the circuit class $\acc$.\footnote{This is the class of languages computed by polynomial size constant-depth circuits consisting of $\gand$, $\gor$, $\gnot$ and $\gmod_m$ gates (for a fixed integer $m \in \mathbb{N}$). Every gate other than $\gnot$ is allowed to have unbounded fan-in.} Moreover, this is the \emph{only} known proof of this result. Other approaches that have been proposed are also based on the design of new $\acc$ algorithms (Chen et al. \citep{DBLP:journals/eccc/ChenKKSZ13}). While his result is still weak compared to the main open problems in circuit complexity and complexity theory in general, it is a landmark in our understanding of the connection between nontrivial algorithms and the existence of hard computational problems.

Can we extend this technique to prove stronger circuit lower bounds? Is there any connection between Williams' transference theorem and other similar results discussed before? This survey is motivated by these questions. We break the presentation into two parts. The first part is a fast-paced introduction to some known results connecting algorithms to circuit lower bounds. The second part of this survey presents complete proofs for most of these theorems, along with some extensions that may be of independent interest. Of course we are not able to cover every result that relates algorithms to circuit lower bounds, but we tried to describe representative results from many areas.

We stress that we focus on \emph{generic} connections between faster algorithms and circuit lower bounds, instead of particular techniques that have found applications in both areas (Fourier representation of boolean functions \citep{DBLP:journals/jacm/LinialMN93}, satisfiability coding lemma \citep{DBLP:journals/cjtcs/PaturiPZ99}, random restriction method \citep{DBLP:journals/eccc/ChenKKSZ13}, etc.). For the reader with basic background in complexity theory, our presentation is essentially self-contained.

\subsection{A summary of some known results}

\subsubsection{Satisfiability algorithms and circuit lower bounds}

The connection between algorithms for hard problems and circuit lower bounds has been known for decades. More precisely, a collapse theorem attributed to Meyer \citep{DBLP:conf/stoc/KarpL80} states that if $\mathsf{EXP} \subseteq \ppoly$ then $\mathsf{EXP} = \Sigma_2^p$ (recall this is the second level of $\mathsf{PH}$, the polynomial time hierarchy). On the other hand, it is not hard to prove that if $\mathsf{P} = \mathsf{NP}$ then $\mathsf{P} = \Sigma_2^p = \mathsf{PH}$. Together, the assumptions that there are efficient algorithms for $\mathsf{NP}$-complete problems and that every problem in $\mathsf{EXP}$ admits polynomial size circuits lead to $\mathsf{P} = \mathsf{EXP}$, a contradiction to the deterministic time hierarchy theorem. In other words, if there exists efficient algorithms for $3$-SAT, it must be the case that $\mathsf{EXP} \nsubseteq \ppoly$.\footnote{Using the fact that $\mathsf{P} = \mathsf{PH}$ implies the collapse of the exponential time hierarchy to $\mathsf{EXP}$, an even stronger consequence can be obtained. We omit the details.} Similar transference results can be obtained from the assumption that there are subexponential time algorithms for $3$-SAT (i.e., with running time $2^{n^{o(1)}}$).

The existence of such algorithms is a very strong assumption. The best known algorithms for $k$-SAT run in time $2^{n(1 - \delta(k))}$, for some fixed constant $\delta(k) > 0$ that goes to zero as $k$ goes to infinity (cf. Dantsin and Hirsch \citep{DBLP:series/faia/DantsinH09}). These algorithms offer an exponential improvement over the trivial running time $\tilde{O}(2^n)$. If we only require the running time to be faster than exhaustive search (``nontrivial''), then improved algorithms are known for many interesting circuit classes (see for instance \citep{DBLP:conf/focs/Santhanam10}, \citep{DBLP:journals/eccc/ChenKKSZ13}, \citep{DBLP:journals/cc/SetoT13}, \citep{DBLP:journals/corr/abs-1212-4548}, \citep{DBLP:conf/coco/BeameIS12}, \citep{DBLP:conf/soda/ImpagliazzoMP12}). For an introduction to some of these algorithms, see Schneider \citep{DBLP:journals/corr/Schneider13}.

It makes sense therefore to consider more refined versions of the transference theorem for satisfiability algorithms. This is precisely the first result in this direction obtained by Williams \citep{DBLP:conf/stoc/Williams10}: the existence of nontrivial algorithms deciding the satisfiability of polynomial size circuits is enough to imply $\nexp \nsubseteq \ppoly$. Unfortunately, $\ppoly$ is a very broad class, and the algorithms mentioned before do not work or have trivial running time on such circuits. 

In follow up work, Williams \citep{DBLP:conf/coco/Williams11} extended his techniques from \citep{DBLP:conf/stoc/Williams10} to prove a more general result that holds for most circuit classes. 

\begin{proposition}[``SAT algorithms yield circuit lower bounds, I'' \citep{DBLP:conf/coco/Williams11}] \label{p:williams1}~\\ Let $\classc$ be a class of circuit families that is closed under composition \emph{(}the composition of two circuits from $\classc$ is also in $\classc$\emph{)} and contains $\ac$. There is a $k>0$ such that, if satisfiability of $\classc$-circuits with $n$ variables and $n^c$ size can be solved in $O(2^n/n^k)$ time for every $c$, then $\nexp \nsubseteq \classc[\poly(n)]$.  
\end{proposition}

In addition, he provided a nontrivial algorithm for $\acc[2^{n^\delta}]$ (the class of $\acc$ circuits of size $2^{n^\delta}$), where $\delta = \delta(d,m) > 0$ depends on the depth of the circuit and the modulo gate. Altogether, these results imply the following circuit lower bound.

\begin{corollary} \label{c:acclb1} $\nexp \nsubseteq \acc$.
\end{corollary}

Subsequent work of Williams \citep{DBLP:conf/stoc/Williams13} has extended these techniques to prove the following stronger transference theorem, which provides better circuit lower bounds\footnote{We use $\ne \cap \iocone$ instead of $\ne \cap \cone$ in the statement of Proposition 
\ref{p:williams2} because the proof described in \citep{DBLP:conf/stoc/Williams13} requires this extra condition \citep{privatecomm}.}.

\begin{proposition}[``SAT algorithms yield circuit lower bounds, II'' \citep{DBLP:conf/stoc/Williams13}] \label{p:williams2}~\\ Let $\classc$ be a class of circuit families that is closed under composition and contains $\ac$. There is a $k>0$ such that, if satisfiability of $\classc$-circuits with $n$ variables and $n^{\log^c n}$ size can be solved in $O(2^n/n^k)$ time for every $c$, then $\ne \cap \iocone \nsubseteq \classc[n^{\log n}]$.  
\end{proposition}

Besides, the following strengthening of Corollary \ref{c:acclb1} is proven in the same paper (the first statement is implicit in his proof).

\begin{corollary}\label{c:acclb2}
$\mathsf{E} \nsubseteq \acc[n^{\log n}]$ or $\mathsf{Quasi}$-$\mathsf{NP} \cap \mathsf{i.o.Quasi}$-$\mathsf{coNP} \nsubseteq \acc[n^{\log n}]$. In particular, $\ne \cap \iocone \nsubseteq \acc[n^{\log n}]$.
\end{corollary}

The proof of these transference theorems has been simplified since then. In Santhanam and Williams \citep{DBLP:journals/eccc/SanthanamW12}, self-reduction (cf. Allender and Kouck\'y \citep{DBLP:journals/jacm/AllenderK10}) is used to obtain an equivalent circuit from a smaller circuit class given an arbitrary $\nc$ circuit (under some assumptions). This simplifies one of the main technical lemmas from \citep{DBLP:conf/coco/Williams11}.

\subsubsection{Constructivity and circuit lower bounds}

There are three significant barriers to circuit lower bound proofs: relativization (Baker, Gill, and Solovay \citep{DBLP:journals/siamcomp/BakerGS75}), natural proofs (Razborov and Rudich \citep{DBLP:journals/jcss/RazborovR97}), and algebrization (Aaronson and Wigderson \citep{DBLP:journals/toct/AaronsonW09}). Roughly speaking, these barriers can be interpreted as follows: some proof methods are too general, and if a lower bound can be obtained by one of such techniques alone, then we get a contradiction to some known result or a widely believed conjecture\footnote{These barriers can also be interpreted as independence results for some formal theories (\citep{arora1992relativizing}, \citep{razborov1995unprovability}, \citep{DBLP:conf/stoc/ImpagliazzoKK09}).}. As explained by Williams \citep{DBLP:conf/coco/Williams11}, his lower bound proof combines several methods used in modern complexity theory, and each one avoids a particular barrier\footnote{We stress however that there is no widely believed conjecture that leads to pseudorandom function families in $\acc$, and this is an interesting open problem. As far as we know, there may exist a natural proof that $\mathsf{P} \nsubseteq \acc$.}. 

It was proven by Razborov and Rudich that most of the circuit lower bound proofs known at the time proceeded (at least implicitly) as follows. There is a circuit class $\classc$ (say, $\ac$) that one wants to separate from a complexity class $\Gamma$ (say, $\mathsf{P}$). In order to do that, one defines a property $\propp$ of boolean functions (i.e., a subset of all boolean functions), and prove that \emph{no} function in $\classc$ satisfies $\propp$, while there exists some hard function $h \in \Gamma$ for which $\propp(h) = 1$ (in this case, we say that $\propp$ is \emph{useful} against $\classc$). For instance, every $\ac$ function simplifies after an appropriate random restriction (\citep{DBLP:journals/mst/FurstSS84}, \citep{DBLP:conf/focs/Yao85}, \citep{haastad1987computational}), while the parity function is still as hard as before. 

As it turns out, for the property $\propp$ defined in these proofs, there is an efficient algorithm $\mathcal{A}$ (with respect to the size of the truth-table of $f$) that is able to decide whether $\propp(f) = 1$. Such properties are referred to as \emph{constructive} properties. In addition, it is usually the case that a random function satisfies $\propp$ with non-negligible probability ($\propp$ satisfies the \emph{denseness} condition). These two conditions imply that $\mathcal{A}$ can be used to distinguish a function in $\classc$ from a random function. Put another way, if there exists a proof of this form that $\Gamma \nsubseteq \classc$, then there is no pseudorandom function family in $\classc$. 

However, if some number-theoretic problems are exponentially hard on average (an assumption believed to be true by many researchers), then there are pseudorandom functions in circuit classes as small as $\tc_4$ (Naor and Reingold \citep{DBLP:journals/jacm/NaorR04}, Krause and Lucks \citep{DBLP:journals/cc/KrauseL01}). As a consequence, such proofs (dubbed \emph{natural proofs} in \citep{DBLP:journals/jcss/RazborovR97}) are not expected to prove separations for more expressive circuit classes. Unfortunately, most (if not all) known combinatorial proofs of circuit lower bounds implicitly define such properties, and this explains the lack of significant progress obtained so far for more general classes of circuits using these techniques only. The interested reader is referred to Chow \citep{DBLP:journals/jcss/Chow11} and Rudich \citep{DBLP:conf/random/Rudich97} for further developments.

As a consequence, any circuit lower bound proof for more expressive classes must violate either the denseness or the constructivity condition. Williams \citep{DBLP:conf/stoc/Williams13} shed light into this problem, by proving that any separation of the form $\nexp \nsubseteq \classc$ is actually equivalent to exhibiting a \emph{constructive} property $\propp$ that is useful against $\classc$. 

\begin{proposition}[``Constructivity is unavoidable'', {\bf informal} \citep{DBLP:conf/stoc/Williams13}]\label{p:constructivity}~\\
Let $\classc$ be a typical circuit class. Then $\nexp \nsubseteq \classc$ if and only if there exists a constructive property $\propp$ that is useful against $\classc$.
\end{proposition}

In other words, any lower bound proof against $\nexp$ implies the existence of a property that is both useful and constructive, but not necessarily dense. As we explain later in the text, $\propp$ is actually computed with a small amount of advice. We clarify this point in Section \ref{s:resultuseful}, where we discuss some additional results about useful properties and circuit lower bounds.

\subsubsection{Additional transference theorems}\label{s:additional}

As alluded to earlier, several additional transference theorems of the form ``faster algorithms yield circuit lower bounds'' have been discovered. In this section we describe some of these results in more detail. We focus on learning algorithms, derandomization, and algorithms for string compression.~\\

\noindent \emph{Derandomization.} There is a strong connection between the existence of pseudorandom generators and circuit lower bounds (see \citep{DBLP:journals/eatcs/Kabanets02}, \citep{DBLP:journals/jcss/Umans03}). Nevertheless, while there is evidence that pseudorandom generators are \emph{necessary} in order to derandomize probabilistic algorithms (Goldreich \citep{DBLP:books/sp/goldreich2011/Goldreich11g}), i.e. to prove that $\mathsf{P} = \mathsf{BPP}$, this is still open. 

On the other hand, for the larger randomized complexity class $\mathsf{MA}$, it is known that any derandomization (such as $\mathsf{MA} \subseteq \mathsf{NSUBEXP}$) implies superpolynomial circuit lower bounds for $\nexp$ (Impagliazzo, Kabanets and Wigderson \citep{DBLP:journals/jcss/ImpagliazzoKW02}). Subsequent work of Kabanets and Impagliazzo \citep{DBLP:journals/cc/KabanetsI04} shows that even the derandomization of a single, specific problem in $\mathsf{BPP}$ leads to some circuit lower bounds. More precisely, let $\mathsf{PIT}$ be the language consisting of all arithmetic circuits that compute the zero polynomial over $\mathbb{Z}$, and $\perm$ be the problem of computing the permanent of integer matrices. We use $\size[\poly]$ to denote the set of languages computed by polynomial size boolean circuits. Similarly, let $\asize[\poly]$ be the family of languages computed by arithmetic circuits of polynomial size over $\mathbb{Z}$.

\begin{proposition}[``Derandomization yields circuit lower bounds'' \citep{DBLP:journals/cc/KabanetsI04}]\label{p:clbfromderandomization}~\\
If $\pit \in \nsubexp$, then at least one of the following results hold:
\begin{itemize}
\item[\emph{(}i\emph{)}] $\nexp \nsubseteq \size[\poly]$; or
\item[\emph{(}ii\emph{)}] $\perm \nsubseteq \asize[\poly]$.
\end{itemize}
\end{proposition}

Aaronson and van Melkebeek \citep{DBLP:journals/toc/AaronsonM11} proved a parameterized version of the result, in addition to showing that $\nexp \cap \conexp$ can be used in place of $\nexp$. Another extension appears in Kinne, van Melkebeek and Shaltiel \citep{DBLP:journals/cc/KinneMS12}.~\\

\noindent \emph{Learning algorithms.} Fortnow and Klivans \citep{DBLP:journals/jcss/FortnowK09} were the first to exhibit a formal connection between learning algorithms and circuit lower bounds. Recall that a learning algorithm $\mathcal{A}$ is given restricted access to a fixed but arbitrary function $f$ from a class of functions $\classc$, and it  should output a hypothesis $h$ that is as close to $f$ as possible. Distinct learning models provide difference access mechanisms to $f$, and impose specific requirements over $h$ ($h$ should be close to $f$, $h \equiv f$, etc.) and $\mathcal{A}$ (learner is randomized, deterministic, etc).

An exact learning algorithm is a \emph{deterministic} algorithm that has access to a membership query oracle $\mathsf{MQ}^f$ and an equivalence query oracle $\mathsf{EQ}^f$, and it is required to output a hypothesis $h$ which agrees with $f$ over all inputs\footnote{On input $x \in \{0,1\}^n$, $\mathsf{MQ}^f(x)$ returns $f(x)$. On input a circuit $c$, $\mathsf{EQ}^f$ outputs ``yes'' if $c \equiv f$, otherwise it outputs an arbitrary input $z$ such that $c(z) \neq f(z)$.}.  

\begin{proposition}[``Learning yields circuit lower bounds'' \citep{DBLP:journals/jcss/FortnowK09}]\label{p:clbfromlearning}~\\
Let $\classc$ be a circuit class. If there exists a subexponential time exact learning algorithm for $\classc$, then $\mathsf{E}^{\mathsf{NP}} \nsubseteq \classc$. 
\end{proposition}

 The original proof of Proposition \ref{p:clbfromlearning} relies on many complexity theoretic results. Subsequent work done by Harkins and Hitchcock \citep{DBLP:conf/icalp/HarkinsH11} strengthened the conclusion to $\mathsf{EXP} \nsubseteq \classc$. Finally, Klivans, Kothari and Oliveira \citep{KKO} used a very simple argument to prove the essentially optimal result that exact learning algorithms for $\classc[s(n)]$ running in time $t(n)$ lead to a circuit lower bound of the form $\mathsf{DTIME}[\poly(t(n))] \nsubseteq \classc[s(n)]$.\footnote{The same result was obtained independently by Russell Impagliazzo and Valentine Kabanets \citep{privatecomm2}.}

For \emph{randomized} learning algorithms, much weaker results are known (a formal definition of the model is discussed in Section \ref{s:learning}). For instance, \emph{efficient} PAC learning algorithms are known to lead to circuit lower bounds against $\mathsf{BPEXP}$, the exponential time analogue of $\mathsf{BPP}$ \citep{DBLP:journals/jcss/FortnowK09}. A slightly stronger result was obtained by Klivans et al. \citep{KKO}, but the underlying techniques do not provide interesting results for randomized subexponential time algorithms. It is an interesting open problem to obtain such extension.~\\

\noindent \emph{Truth-table compression.} More recently, Chen et al. \citep{DBLP:journals/eccc/ChenKKSZ13} considered the problem of designing efficient algorithms that obtain nontrivial compression of strings representing truth-tables from a circuit class $\classc$. In other words, given a string $tt(f_n) \in \{0,1\}^N$, where $f_n:\Boolean$ is a function from $\classc \subseteq \ppoly$ and $N = 2^n$, a compression algorithm should run in time $\poly(N)$ and output a circuit $C$ over $n$ inputs and size $\ll 2^n/n$ such that $tt(C) = tt(f)$. In the same paper, they observed that several circuit lower bounds proofs relying on the method of random restrictions yield efficient compression algorithms. On the other hand, they obtained the following transference theorem.

\begin{proposition}[``Compression leads to circuit lower bounds'' \citep{DBLP:journals/eccc/ChenKKSZ13}]\label{p:clbfromcompression}~\\
Let $\classc$ be a circuit class. Suppose that for every $c \in \mathbb{N}$ there is a deterministic polynomial-time algorithm that compresses a given truth table of an $n$-variate boolean function $f \in \classc[n^c]$ to an equivalent circuit of size $o(2^n/n)$. Then $\nexp \nsubseteq \classc$. 
\end{proposition}

It follows from Proposition \ref{p:clbfromcompression} that designing a compression algorithm for $\acc$ would provide an alternative proof of Corollary \ref{c:acclb1}. This is left as an interesting open problem by \citep{DBLP:journals/eccc/ChenKKSZ13}.

\subsection{A guide to the results discussed in this survey}

\subsubsection{Lower bounds from satisfiability algorithms for low depth circuits}\label{s:resultsat}

Let $\tc_2$ denote the class of polynomial size circuits of depth two with gates corresponding to arbitrary linear threshold functions. As far as we know, it may be the case that $\nexp \subseteq \tc_2$. It makes sense therefore to see if the techniques used in the proof of Corollary \ref{c:acclb1} can be helpful in obtaining a separation against bounded-depth circuit classes of this form. 

A more refined version of Proposition \ref{p:williams1} discussed in \citep{DBLP:conf/coco/Williams11} shows that circuit lower bounds for circuits of depth $d$ follow from satisfiability algorithms for depth $2d + O(1)$. We observe here that it is possible to obtain a tight transference theorem for satisfiability algorithms for constant-depth circuits. Let $\classc_d$ be a circuit class consisting of circuits of depth $d$, and $g$ be an arbitrary function. We write $g[k] \circ \classc_d$ to denote the class of functions computed by circuits of depth $d + 1$ consisting of a top layer gate $g$ of fan-in $k$ that is fed by $k$ circuits from $\classc_d$.

\begin{theorem}[``SAT algorithms for depth $d+2$ yield circuit lower bounds for depth $d$'']\label{t:sat}~\\
Let $\classc$ be a reasonable circuit class. If there exists a nontrivial satisfiability algorithm for $\gand[3] \circ \gor[2] \circ \classc_d[\poly]$, then $\nexp \nsubseteq \classc_d[\poly]$.
\end{theorem}

We define reasonable circuit classes in Section \ref{s:preliminaries}. This result can be obtained through a simple extension of the original technique used by Williams \citep{DBLP:conf/stoc/Williams10}. In particular, our presentation avoids the technical details from \citep{DBLP:conf/coco/Williams11}. A similar theorem is described in Jahanjou, Miles and Viola \citep{DBLP:journals/eccc/Viola13}, but the argument they use is more involved. The proof of Theorem \ref{t:sat} and some additional remarks are presented in Section \ref{s:sat}. 

\subsubsection{Useful properties and circuit lower bounds}\label{s:resultuseful}

We discuss in more detail the relation between circuit lower bounds and useful properties (Proposition \ref{p:constructivity}). Useful properties play a fundamental concept in the context of transference theorems, as explained in the next section.

For a uniform complexity class $\Gamma$ (such as $\mathsf{P}$, $\mathsf{NP}$, etc), we say that a property of boolean functions $\propp$ is a $\Gamma$-property if it can be decided in $\Gamma$. We use $\Gamma/s(m)$ to denote the corresponding complexity class with advice of size $s(m)$, where $m$ is the size of the input. Recall that a property is useful against $\classc$ if it distinguishes some hard function from all functions in $\classc$ (a formal definition is presented in Section \ref{s:preliminaries}). 

First we observe that nondeterminism is of no use in the context of useful properties, which is a somewhat surprising result. The proof of this fact relies on some ideas introduced by Williams \citep{DBLP:conf/stoc/Williams13}.

\begin{theorem}[``$\mathsf{NP}$-property yields $\mathsf{P}$-property'']\label{t:npuseful}~\\
Let $\classc$ be a circuit class. If there exists a $\mathsf{NP}$-property that is useful against $\classc[\poly]$, then there is a $\mathsf{P}$-property that is useful against $\classc[\poly]$.
\end{theorem}

Now we discuss in more detail the connection discovered by Williams (Proposition \ref{p:constructivity}) between constructive useful properties ($\mathsf{P}$-properties under our notation) and circuit lower bounds. It turns out that the statement of Proposition \ref{p:constructivity} requires a broader definition, one for which the algorithm deciding the property is allowed inputs of arbitrary size instead of size $N = 2^n$, where $n \in \mathbb{N}$. Put another way, the algorithm receives any string as input, and is allowed to parse its input size as $2^n + k$. Now it is free to interpret $k$ as an advice string of length $\log N$. We clarify this issue here, and observe that Theorem \ref{t:npuseful} together with standard techniques imply the following characterization of $\nexp$ circuit lower bounds.\footnote{We stress that in this survey any algorithm that decides a property of boolean functions works over strings of length $N = 2^n$, where $n \in \mathbb{N}$.}

\begin{theorem}[``Equivalence between $\nexp$ lower bounds and useful properties'']\label{t:eqnexp}~\\
Let $\classc$ be a circuit class. Then $\nexp \nsubseteq \classc[\poly]$ if and only if there exists a $\mathsf{P}/\log N$-property that is useful against $\classc[\poly]$.
\end{theorem}

It makes sense therefore to investigate whether there exists an equivalence between useful properties computed without advice and circuit lower bounds. The following result holds.

\begin{theorem}[``$\ne \cap \cone$ lower bounds and useful properties'']\label{t:eqnecone}~\\
Let $\classc$ be a circuit class. The following holds: 
\begin{itemize}
\item[\emph{(}i\emph{)}] If $\ne \cap \cone \nsubseteq \classc[\poly]$ then there is a $\mathsf{P}$-property that is useful against $\classc[\poly]$.
\item[\emph{(}ii\emph{)}] If for every $c \in \mathbb{N}$ there exists a $\mathsf{P}$-property that is useful against $\classc[n^{\log^{c}n}]$, then $\ne \cap \iocone \nsubseteq \classc[n^{\log n}]$.
\end{itemize}
\end{theorem}

One direction follows from Theorem \ref{t:npuseful}, while the other is implicit in Williams \citep{DBLP:conf/stoc/Williams13}. Given these results, the following conjecture seems plausible.

\begin{conjecture}[``Equivalence between $\ne \cap \cone$ lower bounds and useful properties?'']\label{c:eqnecone}~\\
Let $\classc$ be a circuit class. Then $\ne \cap \cone \nsubseteq \classc[\poly]$ if and only if there exists a $\mathsf{P}$-property that is useful against $\classc[\poly]$.
\end{conjecture}

We discuss how this conjecture relates to Williams' program for circuit lower bounds in Section \ref{s:stronger}. The results for useful properties are discussed in Section \ref{s:resultuseful} of the survey.

\subsubsection{Applications}\label{s:resultconnections}

It is possible to use the results mentioned in the previous sections to prove the propositions stated in Section \ref{s:additional}. In particular, several transference theorems are in fact \emph{connected}, and improvements in one framework propagates to other results.

The first application that we discuss is for compression algorithms, as investigated by Chen et al. \citep{DBLP:journals/eccc/ChenKKSZ13}. Observe that Proposition \ref{p:clbfromcompression} shows circuit lower bounds for $\nexp$ from exact compression of truth-tables of polynomial size circuits. As mentioned in the same paper, their result can be extended to show that even \emph{lossy} compression algorithms lead to circuit lower bounds. We flesh out the details here.

We say that a circuit class $\classc$ admits \emph{lossy compression} algorithms if there exists an efficient algorithm $\mathcal{A}$ (over inputs of size $N = 2^n$) that when given as input a truth-table $tt(f)$ from $\classc$, where $f: \{0,1\}^n \rightarrow \{0,1\}$, outputs a circuit $C$ of size $o(2^n/n)$ such that $\Pr_x[C(x) = f(x)] \geq .51$. A more general definition is discussed in Section \ref{s:compression}.

\begin{theorem}[``Circuit lower bounds from lossy compression'']\label{t:lossycompression}~\\
Let $\classc$ be a circuit class. The following results hold.
\begin{itemize}
\item[\emph{(}i\emph{)}] If for every $c \in \mathbb{N}$ there exists a lossy compression algorithm for $\classc[n^c]$, then $\nexp \nsubseteq \classc[\poly(n)]$.
\item[\emph{(}ii\emph{)}] If for every $c \in \mathbb{N}$ there exists a lossy compression algorithm for $\classc[n^{\log^c n}]$, then $\ne \cap \iocone \nsubseteq \classc[n^{\log n}]$.
\end{itemize}
\end{theorem}

In other words, any efficient algorithm for lossy compression of strings is either trivial on infinitely many input strings represented by truth-tables from $\tc_2$ (i.e., does not provide a lossy encoding of significantly smaller size), or a new circuit lower bound follows. This result can be obtained as an easy application of Theorem \ref{t:eqnexp}, and its proof is presented in Section \ref{s:compression}.

Next we observe that Proposition \ref{p:clbfromderandomization} (``derandomization yield circuit lower bounds'') follows from the transference theorem for satisfiability. More precisely, Theorem \ref{t:sat} extends to slightly more general algorithms, an observation that we discuss in more detail in Section \ref{s:remarkalgorithm}. Using this generalization, it is possible to prove that if Proposition \ref{p:clbfromderandomization} is false, then a contradiction can be obtained. This proof is presented in Section \ref{s:derandomization}.

In the context of learning algorithms, some extensions of the main result of Fortnow and Klivans \citep{DBLP:journals/jcss/FortnowK09} for exact learning (Proposition \ref{p:clbfromlearning}) follow easily from results for useful properties (Theorems \ref{t:eqnexp} and \ref{t:eqnecone}). In addition, it is not hard to show that even subexponential time randomized learning leads to useful properties decided by efficient randomized algorithms.

\begin{theorem}[``Useful properties and learning algorithms'']\label{t:usefullearning}~\\
Let $\classc = \classc[\poly]$ be a circuit class. If there exists a subexponential time randomized \emph{PAC} learning algorithm for $\classc$, then there exists a $(\mathsf{promise})\mathsf{coRP}$-property that is useful against $\classc$.
\end{theorem}

These transference theorems can be obtained by interpreting learning algorithms as lossy compression schemes, or by relying directly on Theorems \ref{t:eqnexp} and \ref{t:eqnecone}. These results are discussed in more detail in Section \ref{s:learning}.

Overall, these observations show that an improvement of a transference theorem in one framework leads to similar improvements in other frameworks. For instance, a proof of Conjecture \ref{c:eqnecone} implies many interesting results of the form ``nontrivial algorithms yield circuit lower bounds''. More precisely, it immediately implies new transference theorems for both (lossy) compression and satisfiability algorithms, and an alternative proof of the extension of Proposition \ref{p:clbfromderandomization} obtained by Aaronson and van Melkebeek \citep{DBLP:journals/toc/AaronsonM11}. Moreover, a direct improvement of the transference theorems for satisfiability is likely to imply a similar strengthening of Proposition \ref{p:clbfromderandomization}.

\subsubsection{An overview of the results}\label{s:summary}

For convenience of the reader, Figure \ref{fig:summary} summarizes the relations between algorithms and circuit lower bounds discussed in this survey.

\begin{figure}[!h]
\centering
\includegraphics[scale=.75]{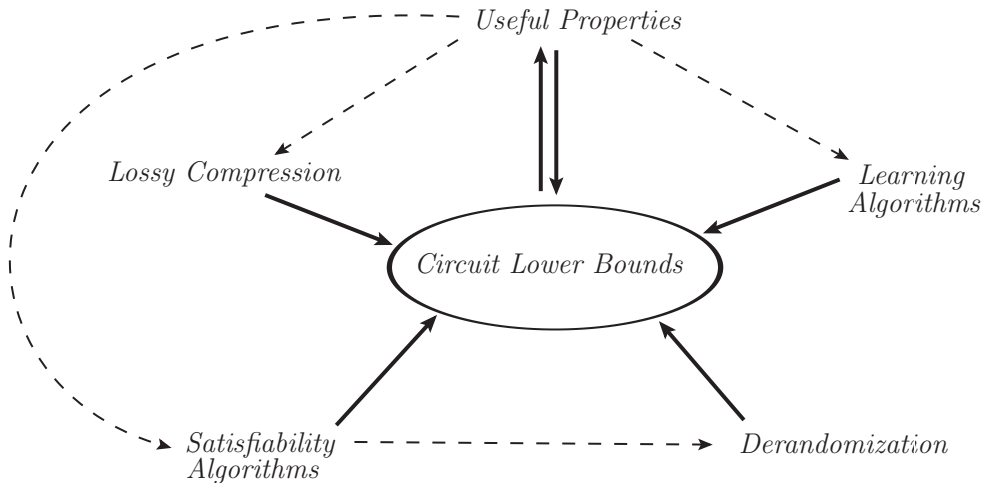}
\caption{\small Bold arrows represent transference theorems, while a dotted arrow from $A$ to $B$ indicates that an improvement of the transference theorem for $A$ implies a similar improvement of the transference theorem for $B$.}
\label{fig:summary}
\end{figure}

\section{Preliminaries and Notation}\label{s:preliminaries}

We assume familiarity with basic notions from computational complexity theory. The reader is referred to Arora and Barak \citep{DBLP:books/daglib/0023084} and Goldreich \citep{DBLP:books/daglib/0019967} for more details. For convenience, we postpone some definitions that are specific to a particular section of the paper to that corresponding section. 

A boolean function is any function $h_n: \Boolean$, for some fixed $n \in \mathbb{N}$. We say that $h$ is a family of boolean functions if $h = \{h_n\}_{n \in \mathbb{N}}$. Any family of functions corresponds to a language $L \subseteq \{0,1\}^*$, and vice versa. For a language $L$, we use $L_n$ to denote $L \cap \{0,1\}^n$. 

We will use $\Gamma$ to denote uniform complexity classes such as $\mathsf{P}$, $\mathsf{coRP}$ and $\mathsf{NP}$. Sometimes we will extend these complexity classes to the corresponding classes with advice of size $s(m)$, where $m$ is the input size. In this case, we use $\Gamma/s(m)$. 
A language $L$ is in $\mathsf{i.o.}\Gamma$ if there exists $L' \in \Gamma$ such that $L_n = L'_n$ for infinitely many values of $n$.

Following \citep{DBLP:conf/stoc/Williams13}, we say that a circuit class $\classc$ is \emph{typical} if $\classc \in \{\ac, \acc, \tc, \nc, \ppoly\}$. The results stated for typical classes hold for more general circuit classes. We use $\size[s(n)]$ to denote the family of functions computed by circuits of size $s(n)$. Similarly, $\asize[s(n)]$ denotes the family of functions computed by arithmetic circuits of size $s(n)$.  Although each circuit class corresponds to a set of languages, we may abuse notation and say that a given circuit $D$ is from $\classc$. In general, for any circuit class $\classc$, let $\classc_d[s(n)]$ be the family of functions computed by circuits from $\classc$ of depth $d$ and size $s(n)$, where the size of a circuit is the number of gates in the circuit. If for convenience we omit $s(n)$, assume the circuits are of polynomial size. For instance, $\tc_2[n^2]$ corresponds to the class of languages computed by circuits of depth-two with $O(n^2)$ gates, each one corresponding to some linear threshold function. All circuit classes considered here are \emph{non-uniform}. If we mention a circuit $D$ of size $s(n)$ without attributing it to a specific circuit class, assume it is composed of $\gand$, $\gor$ and $\gnot$ gates of fan-in at most two.

In order to prove a tight transference theorem for some circuit classes, we make the following definition.

\begin{definition}\label{d:reasonable}
A circuit class $\classc$ is \emph{reasonable} if:
\begin{itemize}
\item[\emph{(}i\emph{)}] The constant zero function $f:\Boolean$ with $f(x) = 0$ for every input $x$ is in $\classc$.
\item[\emph{(}ii\emph{)}] For every function $g \in \classc$, the function $\bar{g} = \gnot(g)$ is in $\classc$, i.e., $\classc$ is closed under complementation. In addition, there is an efficient algorithm that, given the description of a circuit computing $g$, outputs a circuit from $\classc$ of the same size computing $\bar{g}$.
\item[\emph{(}iii\emph{)}] The gates of circuits from $\classc$ may have direct access to constant inputs $0$ and $1$ in addition to the input variables and their negations\footnote{This allows us to hardwire some values without increasing the depth of the circuit.}.
\item[\emph{(}iv\emph{)}] Any language in $\classc[\poly(n)]$ is in $\ppoly$. 
\end{itemize}
\end{definition}

The results that are stated for reasonable classes hold for more general circuit classes, but for simplicity we stick with this definition. In any case, most circuit classes are reasonable (in the sense of Definition \ref{d:reasonable}), including $\ac$, $\tc_2$, $\nc$, $\ppoly$, etc.

We say that a deterministic algorithm is \emph{nontrivial} if it runs in time $2^n/n^{\omega(1)}$. We may use this terminology to talk about nondeterministic and randomized algorithms with similar time bounds.

The following folklore result shows that to prove a circuit lower bound for $\mathsf{P}$ it is enough to obtain a circuit lower bound for the non-uniform class $\ppoly$.

\begin{lemma}\label{l:trick}
Let $\classc_d[\poly(n)]$ be a reasonable circuit class. If $\mathsf{P} \subseteq \classc_d[\poly(n)]$, then for every $b \in \mathbb{N}$ there exists a $t \in \mathbb{N}$ such that every boolean circuit over $n$ inputs of size $n^b$ admits an equivalent circuit from  $\classc_d$ of size $n^t$.
\end{lemma}

\begin{proof}
Assume that $\mathsf{P} \subseteq \classc_d[\poly(n)]$. Consider the following problem:
\begin{center}
$\mathsf{Circuit}$-$\mathsf{Eval}_b = \{ \langle E,x \rangle\;:\; E~\text{is a circuit on}~n~\text{variables of size}\:\leq n^b~\text{and}~E(x)=1\}$
\end{center}
Clearly, $\mathsf{Circuit}$-$\mathsf{Eval}_b$  is in $\mathsf{P}$ (for any fixed $b$), and thus there exists $t$ such that $\mathsf{Circuit}$-$\mathsf{Eval}_b \in \classc_d[n^t]$. In other words, there exists a sequence $\{D_n\}_{n \in \mathbb{N}}$ of circuits from $\classc_d$ of size $O(n^t)$ that computes $\mathsf{Circuit}$-$\mathsf{Eval}_b$.

Let $E_n: \Boolean$ be a function over $n$ boolean variables computed by a circuit of size at most $n^b$. We can hardwire the description of $E_n$ inside circuit  $D_n$ (recall that $\classc$ is reasonable, and that this operation does not increase the depth of the circuit). The resulting circuit is in $\classc_d$, has size at most $n^t$, and it computes $E_n$ by definition of $D_n$.
\end{proof}

The next definition will play an important role in many results discussed later.

\begin{definition}[Properties that are useful against $\classc$ \citep{DBLP:conf/stoc/Williams13}] 
A property of boolean functions is a subset of the set of all boolean functions. For a typical circuit class $\classc$, a property $\propp$ is said to be \emph{useful} against $\classc$ if, for all $k$, there are infinitely many positive integers $n$ such that
\begin{itemize}
\item $\propp(f_n)$ is true for at least one function $f_n:\Boolean$, and
\item $\propp(g_n)$ is false for all functions $g_n:\Boolean$ that admit circuits from $\classc[n^k]$.
\end{itemize}

\noindent We say that $\propp$ is a $\Gamma$-\emph{property} if, given the truth-table $tt(f_n) \in \{0,1\}^N$ \emph{(}where $N = 2^n$\emph{)} of any boolean function $f_n:\Boolean$, $\propp(f_n)$ can be decided in complexity class $\Gamma$. In other words, the language $$L_{\propp} = \{w \in \{0,1\}^N \mid w = tt(f_n)~\textit{for some function}~f_n:\Boolean~\textit{with}~\propp(f_n)=1\}$$ is in $\Gamma$.
\end{definition}

A useful property distinguishes some ``hard'' function from all easy ones. This is weaker than the notion of \emph{natural properties} studied by \citep{DBLP:journals/jcss/RazborovR97}, which also requires $\propp$ to be dense, i.e., $\propp(f)=1$ for a non-negligible fraction of functions.

Recall that a verifier $V$ for a language $L \in \mathsf{NTIME}[t(n)]$ satisfies the following properties:
\begin{itemize}
\item $V(x,w)$ runs in deterministic time $O(t(n))$, where $n = |x|$.
\item $x \in L$ if and only if there exists $w \in \{0,1\}^{O(t(n))}$ such that $V(x,w) = 1$. 
\end{itemize}

\noindent If $L \in \nexp$ and $V$ is a verifier for $L$ running in time $2^{n^{O(1)}}$, we say that $V$ is a $\nexp$-verifier for $L$. Similarly, we may talk about $\ne$-verifiers running in time $2^{O(n)}$.

\begin{definition}
Let $\classc$ be a typical circuit class. We say that a $\nexp$-verifier $V$ for a language $L \in \nexp$ \emph{admits witness circuits from} $\classc[s(n)]$ if for all $x \in L$, there exists a circuit $C \in \classc[s(n)]$ such that $V(x,tt(C)) = 1$. 
\end{definition}

\begin{proposition}[Impagliazzo et al. \citep{DBLP:journals/jcss/ImpagliazzoKW02}, Williams \citep{DBLP:conf/coco/Williams11}] \label{p:witnesscircuits}
Let $\classc$ be a typical circuit class. If $\nexp \subset \classc$ then for any language $L \in \nexp$ and every $\nexp$-verifier $V$ for $L$, there exists $c \in \mathbb{N}$ such that $V$ admits witness circuits from $\classc[n^c]$.
\end{proposition}

\begin{definition}
Given functions $f,g: \Boolean$ and $\delta > 0$, we say that $g$ computes $f$ with advantage $\delta$ if $$\Pr_{x \in_R \{0,1\}^n} [f(x) = g(x)] \geq \frac{1}{2} + \delta.$$
\end{definition}
 
The results for learning algorithms and lossy compression rely on the following fact.

\begin{lemma}[``Random functions are hard to approximate'']\label{l:hardfunction}~\\
There exists a constant $\alpha > 0$ such that for any sufficiently large $n$, there exists a function $h: \Boolean$ that cannot be computed with advantage $\delta > 0$ by any circuit of size $\alpha \cdot 2^n \delta^2 /n$.  
\end{lemma}

\begin{proof}
Fix any circuit $C:\Boolean$. Using the Chernoff-Hoeffding bound, we get that the probability that $C$ computes a random function $r: \Boolean$ with advantage $\delta$ is at most $\exp(-2 \delta^2N)$, where $N = 2^n$ as usual. There are at most $2^{O(s(n) \log s(n))}$ functions on $n$ inputs computed by circuits with $s(n)$ gates. Therefore, it follows by a simple union bound that for some $\alpha > 0$, there exists a function $h$ that is not computed with advantage $\delta$ by any circuit of size $\alpha \cdot 2^n \delta^2 /n$.
\end{proof} 
\section{Lower bounds from SAT algorithms for low depth circuits}\label{s:sat}

In this section we present the transference theorem for satisfiability algorithms. We start with the following definition.

\begin{definition}
Let $\classc$ be a circuit class. We define the computational problem $\equivandc$ as follows. Given the description of circuits from $\classc$ computing functions $f_1,f_2,f_3: \Boolean$, check if $\gand(f_1,f_2)(x) = f_3(x)$ for every $x \in \{0,1\}^n$. The $\equivorc$ problem is defined analogously. 
\end{definition}

\begin{remark}\label{r:equivorc} Observe that if $\classc_d$ is reasonable, then an algorithm for $\equivandc_d$ can be used to solve $\classc_d$-$\mathsf{SAT}$. Moreover, the same algorithm can be used to solve $\equivorc$, since two functions are equivalent if and only if their negations are equivalent \emph{(}by assumption, any reasonable circuit class is closed under negations\emph{)}.
\end{remark}

The proof presented here follows closely the original argument used by Williams \citep{DBLP:conf/stoc/Williams10}, which works for $\ppoly$. However, we introduce a new technique that allows us to obtain an equivalent $\classc$-circuit from a general $\mathsf{P}/\mathsf{poly}$-circuit. It simplifies the proof in \citep{DBLP:conf/coco/Williams11}, and provides a tighter connection between satisfiability algorithms and circuit lower bounds in the case of bounded-depth circuits. The proof of the next lemma is partially inspired by some ideas in Rossman \citep{rossmanphdthesis}.

\begin{lemma}[``Conversion Lemma'']\label{l:translation}
Let $\classc$ be a reasonable circuit class, and suppose that $\mathsf{P} \subseteq \classc$. In addition, assume that there is a nontrivial algorithm for $\equivandc$. Then there exists a nondeterministic algorithm $\mathcal{N}$ with the following properties. Given as input any circuit $B$ over $m$ variables of size $m^b$,  
\begin{itemize}
\item $\mathcal{N}$ has at least one accepting path, and in every accepting path it outputs a circuit $G$ from $\classc[m^t]$ that is equivalent to $B$ \emph{(}where $t = O(b)$\emph{)}.  
\item $\mathcal{N}$ runs in time at most $\frac{2^m}{s(m)}$, for some superpolynomial function $s(m)$.
\end{itemize} 
\end{lemma}

\begin{proof} Let $\mathcal{A}$ be an algorithm for $\equivandc$ running in time $2^m/a(m)$, for a superpolynomial function $a(m)$. 
We proceed as follows. Let $x_1, x_2, \ldots, x_m, g_1, \ldots, g_k$ for $k = m^b$ be a topological sort of the gates of $B$, where each gate $g_i \in \{\gand, \gor, \gnot\}$ has fan-in at most two. We will guess and verify (by induction) equivalent $\classc$-circuits of size $m^t$ for each gate $g_i$ in $B$. Since $\mathsf{P} \subseteq \classc$, it follows from Lemma \ref{l:trick} that the functions computed at the internal gates of $B$ admit such circuit. 

More details follow. Suppose (by induction) that $\mathcal{N}$ has produced equivalent $\classc$-circuits $B^{\classc}_i$ of size at most $m^t$ for every gate $g_i$ of $B$, where $i < \ell$ (otherwise it has aborted already). If $g_\ell$ is an $\gand$ gate with inputs $g_{i_1}, g_{i_2}$, where $i_1, i_2 < \ell$,  $\mathcal{N}$ guesses a circuit $B^{\classc}_\ell$ in $\classc[m^t]$ over the same input variables, then use the $\equivandc$ algorithm to check if $\gand(B^{\classc}_{i_1},B^{\classc}_{i_2})$ and $B^{\classc}_\ell$ are equivalent. $\mathcal{N}$ rejects if these circuits are not equivalent, otherwise it continues the computation, completing the induction step. If $g_\ell$ corresponds to an $\gor$ gate, a similar computation is performed, this time applying an algorithm for $\equivorc$ (check Remark \ref{r:equivorc}). Finally, if $g_\ell$ is a $\gnot$ gate, using the fact that $\classc$ is reasonable, it is possible to produce in polynomial time an equivalent $\classc$-circuit for $g_\ell$ of the same size. This completes the induction step. Observe that the base case is trivial.

Note that $\mathcal{N}$ runs algorithm $\mathcal{A}$ for at most $k$ times, i.e., a polynomial number of times. In addition, each execution is performed over circuits from $\classc$ of size $O(m^t)$. Therefore the total running time of $\mathcal{N}$ is $\poly(m) \cdot 2^m/a(m)$, for some superpolynomial function $a(m)$. Setting $s(m) = a(m)/\poly(m)$ completes the proof of Lemma \ref{l:translation}. 
\end{proof}

In addition, we will need the following auxiliary results, whose notation we borrow from Williams \citep{DBLP:conf/coco/Williams11}.

\begin{definition}\label{d:ssat}
The computational problem $\ssat$ is defined as follows. Given a circuit $C$ over $n$ input variables, denote by $F_C$ the instance of $3$-\emph{SAT} obtained by evaluating $C$ over all inputs in lexicographic order \emph{(}i.e., $F_C$ is the $2^n$-bit string representing the truth-table $tt(C)$ of $C$\emph{)}. Decide if $F_C$ is satisfiable.
\end{definition}

We say that $F_C$ is the \emph{decompression} of $C$, and call $C$ the \emph{compression} of $F_C$. 

\begin{lemma}[Tourlakis \cite{DBLP:journals/jcss/Tourlakis01}, Fortnow et a. \cite{DBLP:journals/jacm/FortnowLMV05}, Williams \cite{DBLP:conf/stoc/Williams10}] \label{l:reduction} There is a fixed constant $c>0$ for which the following holds. For every $L \in \ntime[2^n]$ there is a polynomial time reduction from $L$ to $\ssat$ that maps every input $x$ of size $n$ to a circuit $C_x$ over at most $n + c\log n$ input variables and size $O(n^c)$, such that $x \in L$ if and only if the decompressed formula $F_{C_x}$ is satisfiable \emph{(}observe that this is a formula of size $2^n \poly(n)$\emph{)}. 
\end{lemma}

\begin{definition}\label{d:succinctassign}
We say that $\ssat$ admits \emph{succinct satisfying assignments} if there exists a constant $c>0$ such that for every language $L \in \ntime[2^n]$ the following holds. Given any $x \in L$, there exists some circuit $W_x$ of polynomial size over $k \leq n + c \log n$ input variables for which the assignment $z_i = W(i)$ for $i \in  \{1, \ldots, 2^k\}$ is a satisfying assignment for $F_{C_x}$, where $C_x$ is the circuit obtained from the reduction to $\ssat$ given by Lemma \ref{l:reduction}. 
\end{definition}

The following lemma is an easy consequence of Proposition \ref{p:witnesscircuits}.

\begin{lemma}\label{l:witness} If $\nexp \subseteq \ppoly$ then $\ssat$ admits succinct satisfying assignments.
\end{lemma}

We use these auxiliary results to prove the following proposition. For simplicity, we only state it for polynomial size classes, but a parameterized version can be obtained using the same techniques. 

\begin{prop}\label{p:main}
Let $\classc = \classc_d[\poly(n)]$ be a reasonable circuit class. If there exist a nontrivial algorithm for $\equivandc$, then $\nexp \nsubseteq \classc$.
\end{prop}  

\begin{proof}
Let $\mathcal{A}$ be a nontrivial algorithm for $\equivandc$, and assume for the sake of a contradiction that $\nexp \subseteq \classc$. We use these assumptions to show that every language $L \in \ntime[2^n]$ is in $\ntime[o(2^n)]$, a contradiction to the nondeterministic time hierarchy theorem (\citep{DBLP:journals/jcss/Cook73}, \citep{DBLP:journals/jacm/SeiferasFM78}, \citep{vzak1983turing}). 

The proof relies on the fact that every language $L \in \ntime[2^n]$ can be efficiently reduced to an instance of the $\ssat$ problem (Lemma \ref{l:reduction}). In other words, there is a polynomial time algorithm that maps any input $x \in \{0,1\}^n$ to a circuit $D_x$ on $n + c \log n$ input variables and at most $O(n^c)$ gates such that $x \in L$ if and only if the decompression $F_x = tt(D_x)$ of $D_x$ is satisfiable.   

It follows from $\nexp \subseteq \classc \subseteq \ppoly$ ($\classc$ is reasonable) and Lemma \ref{l:witness} that if $F_x$ is satisfiable then there is a satisfying assignment encoded by a circuit $E_x$ of polynomial size over $n + O(\log n)$ variables. Summarizing what we have so far:
\begin{center}
$x \in L \quad \Longleftrightarrow \quad \exists$ circuit $E:\{0,1\}^{n + O(\log n)} \rightarrow \{0,1\}$ of size $O(n^d)$ such that $F_x(tt(E))=1$, 
\end{center}
\noindent where $F_x = tt(D_x)$ is a $3$-CNF formula and $D_x:\{0,1\}^{n + O(\log n)} \rightarrow \{0,1\}$ is an arbitrary circuit (not necessarily in $\classc$) of  size $O(n^c)$ encoding this formula.

Our nontrivial algorithm for $L$ now guesses a candidate circuit $E$ of this form. It uses $D_x$ and three copies of $E$ to build a circuit $B = B(D_x,E)$ of size $O(n^b)$ over $n + O(\log n)$ inputs such that: 
\begin{center}
$B$ is satisfiable $\quad \Longleftrightarrow \quad$ some clause $C_i$ of $F_x$ is not satisfiable by the assignment $tt(E)$.
\end{center}

The description of $B$ is as follows.  An input $y$ to $B$ is interpreted as an integer $i$, and $B$ uses this index to obtain from $D_x$ the description of the $i$-th clause $C_i$ in $F_x$. Let $z_1, z_2, z_3$ be the literals in $C_i$. Circuit $B$ uses three copies of $E$ to obtain the boolean values of the variables corresponding to these literals, and finally outputs $1$ if and only if these values do not satisfy $C_i$. This last verification can be done by a polynomial size circuit. Observe that $B$ is not a $\classc$-circuit: $D_x$ and $E$ are arbitrary circuits, these circuits are composed, and there is additional circuitry computing the final output value of $B$. Overall, we obtain:
\begin{equation}\label{e:equivalence}
x \in L \quad \Longleftrightarrow \quad~\text{circuit}~B:\{0,1\}^{n + O(\log n)} \rightarrow \{0,1\}~\text{is unsatisfiable}.
\end{equation}

Note that all these steps can be performed in $\ntime[\poly(n)]$. Recall that we can use $\equivandc$ to solve $\classc$-$\mathsf{SAT}$ in less than $2^n$ steps, but $B$ is not a circuit from $\classc$. We can assume without loss of generality that $B$ is a circuit of size $m^b$ (where $m = n + O(\log n)$) consisting of $\gand$, $\gor$ and $\gnot$ gates of fan-in at most two.

While in Williams' original proof there is a step that guesses and verifies an equivalent $\classc$-circuit for $D_x$ (and already assumes $E$ in $\classc[\poly(n)]$ with some extra work), our nondeterministic algorithm for $L$ produces directly an equivalent $\classc$-circuit for the final circuit $B$. Under our assumptions, Lemma \ref{l:translation} can be applied, and it allows the nondeterministic algorithm for $L$ to obtain a circuit $G$ over $m$ inputs from $\classc[m^t]$ that is equivalent to $B$. This step can be performed in time $2^m/s(m)$ for a superpolynomial function $s(m)$. Since $m = n + O(\log n)$, this running time is still nontrivial in $n$.

Using condition (\ref{e:equivalence}), it follows that $x \in L$ if and only if $G$ is unsatisfiable. Finally, since $\classc$ is reasonable, we can use  algorithm $\mathcal{A}$ to check if this is true, in which case our algorithm for $L$ accepts input $x$. Again, this is a computation that can be performed in nontrivial running time by our assumption over $\mathcal{A}$. Overall, it follows that we can decide $L$ in $\ntime[o(2^n)]$, which completes the proof of the theorem.
\end{proof}

\begin{corollary}
Let $\classc = \classc_d[\poly(n)]$ be a reasonable circuit class. If there exist nontrivial satisfiability algorithms for both $\gand[3] \circ \classc$ and $\gand[2] \circ \gor[2] \circ \classc$, then $\nexp \nsubseteq \classc$. 
\end{corollary}

\begin{proof}
It is enough to observe that these satisfiability algorithms can be used to solve $\equivandc$ in nontrivial running time (Proposition \ref{p:main}). Let $f_1, f_2, f_3$ be functions from $\classc$. Then
\begin{equation}\label{e:equivxor}
\neg \gequiv(\gand(f_1,f_2),f_3)\quad \Longleftrightarrow \quad  \gxor(\gand(f_1,f_2),f_3)~\text{is satisfiable}.
\end{equation}  
For bits $a,b \in \{0,1\}$, we have $\gxor(a,b) \equiv \gor(\gand(a, \bar{b}),\gand(\bar{a},b))$. Using de Morgan's rules and combining gates, it is not hard to see that
$$
\gxor(\gand(f_1,f_2),f_3) \equiv \gor(\gand(f_1,f_2,\bar{f_3}),\gand(\gor(\bar{f_1}, \bar{f_2}),f_3)).
$$
It follows from (\ref{e:equivxor}) that an algorithm for $\neg \equivandc$ should output $1$ if and only if either $\gand(f_1,f_2,\bar{f_3})$ or $\gand(\gor(\bar{f_1}, \bar{f_2}),f_3)$ is satisfiable. Since $\classc$ is reasonable, a circuit for $\bar{f_i}$ can be computed efficiently from a circuit for $f_i$. Hence nontrivial algorithms for $\gand[3] \circ \classc$-SAT and $\gand[2] \circ \gor[2] \circ \classc$-SAT can be used to solve $\neg \equivandc$, which completes the proof.
\end{proof}

\subsection{A remark for the algorithm designer}\label{s:remarkalgorithm}

It is hard to find satisfiability algorithm for expressive circuit classes even when we allow very modest running times, such as $2^n/n^{\log n}$. Here we  mention a weaker assumption on the algorithmic side that could be of practical significance\footnote{This is not the most encompassing definition, but it is a very natural one to have in mind.}.

\begin{definition}[``Algorithms useful for circuit lower bounds'']\label{d:usefulalgorithms}~\\
Let $\classc$ be a circuit class. A nondeterministic algorithm $\mathcal{A}$ for $\equivandc$ is \emph{useful} if the following conditions hold:
\begin{itemize}
\item Every path of the \emph{(}nondeterministic\emph{)} computation of $\mathcal{A}$ either outputs ``abort'', or provides the correct answer.
\item At least one path of the computation of $\mathcal{A}$ does not abort, and runs in time bounded by $2^n/s(n)$ for some superpolynomial function $s(n)$.
\end{itemize}
\end{definition}

\begin{proposition}\label{p:useful}
Let $\classc = \classc_d[\poly(n)]$ be a reasonable circuit class. If there exists a useful algorithm for $\equivandc$ then $\nexp \nsubseteq \classc$.
\end{proposition}

\begin{proof}
Observe that the proof of Proposition \ref{p:main} still holds with such algorithms, provided that we abort in any computation path that runs for more than $2^n/s(n)$ steps. It is still the case that $x \in L$ if and only if there exists a computation path that accepts $x$. More precisely, if $x \notin L$, even if equivalent circuits are guessed and verified in each stage, a useful algorithm for $\equivandc$ will never output ``yes'' in the last step of the computation that checks if the final circuits is equivalent to the zero function (i.e., it is unsatisfiable). On the other hand, for $x \in L$, it is clear from the definition of useful algorithm that some computation path will accept in nontrivial running time. 
\end{proof}

Observe that useful algorithms for \emph{unsatisfiability} also lead to circuit lower bounds, since these can be used in place of an $\equivandc$ algorithm. The same is true for satisfiability algorithms, since useful algorithms are closed under complementation. In Section \ref{s:derandomization} we will use Proposition \ref{p:useful} to prove that derandomization implies circuit lower bounds (Proposition \ref{p:clbfromderandomization}).

Why is this a natural relaxation? Suppose there exists a class $\classc$ such that for any circuit $D$ in this class, there exists some subset $S \subset [n]$ of the inputs of $D$ such that by trying all assignments to the variables in $S$, we can check on \emph{average time} strictly less than $2^{n - |S|}$ (over the restrictions) the satifiability of the remaining circuits. Then $\classc$ admits a useful satisfiability algorithm, since the set $S$ can be guessed at the beginning of the execution. For the reader familiar with the satisfiability algorithm for small threshold circuits described by Impagliazzo, Paturi and Schneider \citep{DBLP:journals/corr/abs-1212-4548}, it means that their algorithm gives more than what is needed for lower bounds. There the expected running time is nontrivial over the subset of inputs to be restricted, which is a stronger guarantee. It is sufficient that a single subset provides a nontrivial running time.

\section{Useful properties and circuit lower bounds}\label{s:useful}

In this section we focus on the relation between \emph{useful properties} and circuit lower bounds, a connection that was made explicit in a recent paper written by Williams \citep{DBLP:conf/stoc/Williams13}. We start with the following simple, but somewhat surprising result. Recall that an algorithm that computes a property of boolean functions receives as input a string of size $N = 2^n$ representing the truth-table $tt(f)$ of a function $f:\Boolean$. The following result follows from techniques introduced by Williams \citep{DBLP:conf/stoc/Williams13}.

\begin{proposition}[``Useful $\mathsf{NP}$-property yields useful $\mathsf{P}$-property''] \label{p:npuseful}~\\
Let $\classc$ be a typical circuit class, and let $s: \mathbb{N} \rightarrow \mathbb{N}$ be any function. If there is a $\mathsf{NP}/s(N)$-property useful against $\classc$ then there is a $\mathsf{P}/s(N)$-property useful against $\classc$. 
\end{proposition}

\begin{proof} First we prove the proposition without advice, then we observe that the same proof works in the presence of advice strings as well.
Let $\propp$ be a $\mathsf{NP}$-useful property against $\classc$. In other words, for any fixed $k$, there exists an infinite subset $S_k \subseteq \mathbb{N}$ such that for any $n \in S_k$:
\begin{itemize}
\item $\propp(f_n) = 1$ for at least one function $f_n:\Boolean$.
\item $\propp(g_n) = 0$ for any function $g_n: \Boolean$ computed by circuits in $\classc[n^k]$.
\end{itemize}
In addition, there exists a polynomial time verifier $V_\propp:\{0,1\}^N \times \{0,1\}^{N^c - N} \rightarrow \{0,1\}$ (where $N = 2^n$ and $c \in \mathbb{N}$) for $L_\propp$. Put another way, for any function $h_n$, 
$$\propp(h_n) = 1\quad\quad\Longleftrightarrow\quad\quad \exists w \in \{0,1\}^{N^c - N}~\text{such that}~~V_\propp(tt(h_n),w)=1.
$$

Let $A = \{n' \mid n' = cn, n \in \mathbb{N}\}$. For convenience, set $N' = 2^{n'} = N + (N^c - N)$. We define a predicate $\propp'$ defined on any function over $n'$ inputs, where $n' \in A$ (the definition of $\propp'$ over functions with a different number of inputs can be arbitrary). For any $h'_{n'}: \{0,1\}^{n'} \rightarrow \{0,1\}$, view its representation $tt(h'_{n'}) \in \{0,1\}^{N'}$ as a pair of strings $(tt(h_n), w)$, where $tt(h_n) \in \{0,1\}^N$ and $w \in \{0,1\}^{N^c - N}$. To be more precise, let $h_n: \Boolean$ be the restriction of $h'_{n'}$ defined by $h_n(x) = h'_{n'}(x0^{(c-1)n})$, where $x \in \{0,1\}^n$. Finally, let
$$ \propp'(h'_{n'}) = 1  \quad\quad\Longleftrightarrow\quad\quad V_\propp(tt(h_{n}),w)=1.$$

We claim that $\propp'$ is a $\mathsf{P}$-property that is useful against $\classc$. First observe that since $V_\propp$ is an efficient algorithm, $\propp'$ can be computed in time polynomial in $N' = |tt(h'_{n'})|$. Fix any $k \in \mathbb{N}$. We need to define an infinite set $S'_{k} \subseteq A$ such that for every $n' \in S'_{k}$,
\begin{itemize}
\item $\propp(f'_{n'}) = 1$ for at least one function $f'_{n'}: \{0,1\}^{n'} \rightarrow \{0,1\}$.
\item $\propp(g'_{n'}) = 0$ for any function $g'_{n'}: \{0,1\}^{n'} \rightarrow \{0,1\}$ computed by circuits in $\classc[n'^{k}]$.
\end{itemize}

\noindent Let $S'_{k} = \{n' \mid n' = cn, n \in S_{k+1}\}$. This set is infinite because so is $S_{k+1}$. Let $n' \in S'_{k}$. It follows from the definition of $S_{k+1}$ that there is a function $f_n: \Boolean$ for which $\propp(f_n) = 1$. Hence there exists $w \in \{0,1\}^{N^c - N}$ such that $V_\propp(tt(f_n),w) = 1$. By construction, the corresponding function $f'_{n'}: \{0,1\}^{n'} \rightarrow \{0,1\}$ whose truth-table is the concatenation of the pair $(f_n,w)$ satisfies $\propp'$. 

Finally, in order to establish the second bullet, assume for the sake of a contradiction that there exists a function $g'_{n'}: \{0,1\}^{n'} \rightarrow \{0,1\}$ computed by circuits from $\classc[n'^{k}]$ for which $\propp'(g'_{n'}) = 1$. Clearly, the function $g_n: \Boolean$ defined as before by the restriction $g_n(x) = g'_{n'}(x0^{(c-1)n})$ also admits circuits from $\classc$ of size $n'^{k} = (cn)^{k} \leq n^{k+1}$, for sufficiently large values of $n$. But then $\propp(g_n) = 0$, since $n \in S_{k+1}$. However, this contradicts the assumption that $\propp'(g'_{n'}) = 1$, since in this case there is no $w \in \{0,1\}^{N^c - N}$ such that $V_\propp(tt(g_n),w) = 1$. In other words, for every function $g'_{n'}$ with $n' \in S'_{k}$ that is computed by circuits from $\classc[n'^{k}]$, we have $\propp'(g'_{n'}) = 0$. 

If the original verifier works with advice strings of length $s(N)$, then property $\mathcal{P}'$ can be decided correctly using the same advice. However, the definition of $\mathcal{P}'$ over functions on $n' = cn$ inputs is based on the definition of $\mathcal{P}$ over functions on $n$ inputs. Therefore, the advice for the new algorithm is of size $s(N'^{1/c})$, since it gets as input truth-tables of size $N' = N^{c}$. Assuming that $s(.)$ is non-decreasing and $c \geq 1$, it follows that $\mathcal{P}'$ can be decided with advice of size $s(N'^{1/c}) \leq s(N')$. This completes the proof of Proposition \ref{p:npuseful}. 
\end{proof}

The new useful property may not be dense, even if the original property is dense. The reason is that there may be just a few certificates for each hard function, thus almost no function will satisfy the newly defined property. However, if we start with an $\mathsf{RP}$-natural property useful against $\classc$ (i.e., a dense property in which every hard function has many certificates), the proof of Proposition \ref{p:npuseful} yields a corresponding $\mathsf{P}$-natural property.

The next proposition clarifies the relation between $\nexp$ circuit lower bounds and the existence of properties that are useful against $\classc$. Recall that for any typical circuit class, standard arguments can be used to prove that $\nexp \nsubseteq \classc$ if and only if $\ne \nsubseteq \classc$.

\begin{proposition}\label{p:usefullog}
Let $\classc$ be a typical class. Then $\nexp \nsubseteq \classc$ if and only if there exists a $\mathsf{P}/\log N$-property that is useful against $\classc$.
\end{proposition}

\begin{proof} Let $N = 2^n$ as usual. First assume that $\nexp \nsubseteq \classc$, and let $L \in \mathsf{NE} \backslash \classc$. Let $L' = L \cup \{1^n\mid n \in \mathbb{N}\}$, and notice that $L' \in \mathsf{NE} \backslash \classc$. For every $n \in \mathbb{N}$, let $b(n)$ be the number of strings of size $n$ in $L'$. Observe that $b(n) \in [1,2^n]$. Therefore $b(n)$ can be encoded by a string $a(n) \in \{0,1\}^{\log N}$. Let $f_n = L_n$, i.e., $f_n(x) = 1$ if and only if $x \in L$. Consider the property $\mathcal{P}$ such that $\mathcal{P}(g)=1$ if and only if $g = f_n$ for some $n \in \mathbb{N}$. We claim that $\mathcal{P}$ is a $\mathsf{NP}/\log N$-property that is useful against $\classc$. Let $V'$ be an $\mathsf{NE}$-verifier for $L'$ accepting witnesses of size $2^{cn}$.

Clearly, $\mathcal{P}$ is useful against $\classc$, because $L' \notin \classc$. On the other hand, the following $\mathsf{NP}$-verifier decides $\mathcal{P}$ when it is given the correct advice string $a(n)$:\\

\noindent {\bf Verifier $V$ for $\mathcal{P}$:}\\
\noindent On inputs $tt(h) \in \{0,1\}^N$ and advice string $z \in \{0,1\}^{\log N}$, reject if $|h^{-1}(1)| \neq z$. Otherwise, guess witnesses $w_x \in \{0,1\}^{N^c}$ for every $x \in h^{-1}(1)$, and accept if and only if $V'(x,w_x)=1$ for every such $x$.\\

Clearly, when $z = a(n)$, the only function over $n$ inputs accepted by $V$ is $f_n = L_n$. In addition, $V$ runs in time $\mathsf{poly}(N)$. It follows that $\mathcal{P}$ is computed in $\mathsf{NP}/\log N$. Therefore, there is a $\mathsf{NP}/\log N$-property that is useful against $\classc$, and Proposition \ref{p:npuseful} guarantees the existence of a $\mathsf{P}/\log N$-property useful against $\classc$.

Now suppose that there exists a $\mathsf{P}/\log N$-property $\mathcal{P}'$ that is useful against $\classc$. We use this assumption to define a $\nexp$-verifier $V'$ that does not admit witness circuits of polynomial size. Observe that it follows then from Proposition \ref{p:witnesscircuits} that $\nexp \nsubseteq \classc$, which completes the proof our result.

Let $\mathcal{A}'$ be an algorithm running in time $N^d$ that decides $\mathcal{P}'$ on inputs $tt(f) \in \{0,1\}^N$ when it is given access to an appropriate advice string $a(N) \in \{0,1\}^{\log N}$, i.e, a string of size $n$. Consider the following verifier.\\

\noindent {\bf $\nexp$-verifier $V'$:}\\
\noindent On input $\langle x,w \rangle$, where $x \in \{0,1\}^n$ and $w \in \{0,1\}^N$, output $\mathcal{A}'(w)/x$ (i.e., run $\mathcal{A}'$ on input $w$ with advice string $x$).\\

First observe that $V'$ is a $\nexp$-verifier. Fix any $c \in \mathbb{N}$. We prove that $V'$ does not admit witness circuits from $\classc[n^c]$. First, there are infinitely many inputs $n$ for which $\mathcal{P'}$ correctly discriminates a hard function $h_n: \{0,1\}^n \rightarrow \{0,1\}$ from a function in $\classc[n^c]$. For any such value of $n$, there is a correct advice string $a(N)$ for which algorithm $\mathcal{A}'$ computes $\mathcal{P}'$. However, whenever $x =  a(N)$, it follows from the definition of $V'$ that it only accepts certificates for $x$ that do not correspond to any truth-table from $\classc[n^c]$. In addition, $V'$ accepts at least one truth-table, by definition of $\mathcal{P}'$. As discussed before, this completes the proof of Proposition \ref{p:usefullog}.
\end{proof}

One may be tempted to pose the following conjecture.

\begin{conjecture}\label{c:advice}
Let $\classc$ be a typical circuit class. If there exists a $\mathsf{P}/O(\log N)$-property that is useful against $\classc$, then there is a $\mathsf{P}$-property that is useful against $\classc$.
\end{conjecture}

We will see shortly that if a slightly more general version of this conjecture holds, then a generic $\nexp$ circuit lower bound can always be converted into a $\ne \cap \cone$ lower bound, a rather surprising consequence, given its generality. 

Now we move to the relation between useful properties decided without advice and circuit lower bounds.

\begin{proposition}\label{p:neconeuseful}
For any typical $\classc$,  if $\ne \cap \cone \nsubseteq \classc$ then there exists a $\mathsf{P}$-property that is useful against $\classc$.
\end{proposition}

\begin{proof}
Let $L \in \ne \cap \cone \backslash \classc$, and let $V^0$ and $V^1$ be verifiers running in time $2^{O(n)}$ for $n = |x|$ such that:
\begin{eqnarray}
x \in L \quad & \Longleftrightarrow & \quad \exists w_x \in \{0,1\}^{2^{O(n)}}~\text{such that}~\:V^1(x,w_x)=1. \nonumber \\
x \notin L \quad & \Longleftrightarrow & \quad \exists w_x \in \{0,1\}^{2^{O(n)}}~\text{such that}~\:V^0(x,w_x)=1. \nonumber
\end{eqnarray}
\noindent We view $L$ as a family of functions $f = \{f_n\}_{n \in \mathbb{N}}$, where $f_n^{-1}(1) = L_n$. Let $\mathcal{P} = \{f_n \mid n \in \mathbb{N}\}$. First observe that this property is useful against $\classc$, since $L \notin \classc$. In addition, there is an efficient verifier $V_\propp$ for $\propp$: on input a string $tt(h) \in \{0,1\}^N$ representing the truth-table of a function $h:\Boolean$, guess $2^n$ certificates $y_x \in \{0,1\}^{N^c}$, one for each $x \in \{0,1\}^n$, and accept if and only if $V^{h(x)}(x,y_x) = 1$ for every such $x$. Clearly, $V_\propp$ is a $\mathsf{NP}$-verifier for $\propp$. It follows then from Proposition \ref{p:npuseful} that there exists a $\mathsf{P}$-property $\propp'$ that is useful against $\classc$, which completes the proof.   
\end{proof}

Conversely, which consequences can we obtain from the existence of $\mathsf{P}$-properties (without advice) that are useful against $\classc$? The following result is implicit in the work of Williams \citep{DBLP:conf/stoc/Williams13}, and it shows that without advice even stronger consequences can be obtained (although in the quasipolynomial size regime). 

\begin{proposition}\label{p:unary}
Let $\classc$ be a typical circuit class. If for every $c \in \mathbb{N}$ there exists a $\mathsf{P}$-property that is useful against $\classc[n^{\log^{c}n}]$, then $\mathsf{NE} \cap \iocone \nsubseteq \classc[n^{\log n}]$.
\end{proposition}

\noindent We give a self-contained proof of this result in Appendix \ref{a:unary}.

Proposition \ref{p:unary} sheds some light into Conjecture \ref{c:advice}. It shows that if the analogue of this conjecture for quasipolynomial size circuits holds, then $\nexp$ lower bounds against such circuits can be translated into similar $\ne \cap \cone$ circuit lower bounds (via a generalization of Proposition \ref{p:usefullog} to quasipolynomial size circuits). 

Given the statement of Propositions \ref{p:neconeuseful} and \ref{p:unary}, it is plausible to conjecture that there is a tight correspondence between useful properties computed without advice and circuit lower bounds for $\ne \cap \cone$. 

\begin{conjecture}\label{c:unary}
Let $\classc = \classc[\poly]$ be a typical circuit class. Then $\ne \cap \cone \nsubseteq \classc$ if and only if there exists a $\mathsf{P}$-property that is useful against $\classc$.
\end{conjecture}

We will see in Section \ref{s:applications} that useful properties are powerful enough to simplify and generalize many results of the form ``nontrivial algorithms yield circuit lower bounds''. In particular, a proof of Conjecture \ref{c:unary} would provide stronger transference theorems in different frameworks.

\subsection{Satisfiability algorithms and useful properties}\label{s:stronger}

It is possible to formulate the main result from Section \ref{s:sat} as follows: the existence of nontrivial satisfiability algorithms leads to useful properties, which in turn imply circuit lower bounds. This can be accomplished using the fact that the nondeterministic hierarchy theorem also holds for unary languages. In other words, if there exists a nontrivial SAT algorithm for a circuit class $\classc$, the proof of Proposition \ref{p:main} shows that any verifier for a hard unary language must have infinitely many inputs that only admit certificates of high $\classc$-circuit complexity. This verifier can be used to define a property that is useful against $\classc$: given a truth table $tt(h_n)$, check if it is a valid certificate for the input $1^n$. 

More specifically, satisfiability algorithms for polynomial size circuits lead to $\mathsf{P}$-properties useful against circuits of polynomial size, while algorithms for quasipolynomial size circuits lead to $\mathsf{P}$-properties useful against circuits of such size. The reader should compare the transference theorems from \citep{DBLP:conf/coco/Williams11} and \citep{DBLP:conf/stoc/Williams13} (Propositions \ref{p:williams1} and \ref{p:williams2}, respectively) to the statements of Propositions \ref{p:usefullog} and \ref{p:unary}. If Conjecture \ref{c:unary} is true, the existence of nontrivial satisfiability algorithms for $\classc[\poly]$ would imply that $\ne \cap \cone \nsubseteq \classc[\poly]$, a new result.

\section{Applications}\label{s:applications}

\subsection{Lower bounds from lossy compression}\label{s:compression}

In this section we prove the transference theorem obtained by Chen et al. \citep{DBLP:journals/eccc/ChenKKSZ13}, which we state again for convenience.

\begin{proposition}[Compression yields circuit lower bounds \citep{DBLP:journals/eccc/ChenKKSZ13}]~\\
Let $\classc$ be a typical circuit class. Suppose that for every $c \in \mathbb{N}$ there is a deterministic polynomial-time algorithm that compresses a given truth table of an $n$-variate boolean function $f \in \classc[n^c]$ to an equivalent circuit of size $o(2^n/n)$. Then $\nexp \nsubseteq \classc$. 
\end{proposition}

As mentioned before, it is possible to show a similar result from the existence of lossy compression algorithms.
  
\begin{definition}[Lossy compression scheme] Let $\classc$ be a typical circuit class. We say that a \emph{deterministic} algorithm $\mathcal{A}$ is a $(\delta(n),s(n))$-\emph{compression algorithm} for $\classc$ if $\mathcal{A}$ runs in time $\poly(N)$, and for any fixed $k \in \mathbb{N}$, there are infinitely many integers $n$ for which the following holds. Given any string $tt(f_n) \in \{0,1\}^N$ representing a function $f_n: \Boolean$ computed by circuits in $\classc[n^k]$, $\mathcal{A}$ outputs a circuit $C$ on $n$ inputs of size at most $s(n)$ that computes $f_n$ with advantage $\delta(n)$.
\end{definition}  

\begin{proposition}[Lossy compression yields circuit lower bounds] \label{p:clbfromapprox}~\\
Let $\classc$ be a typical circuit class, and let $\delta(n): \mathbb{N} \rightarrow (0,1/2]$ be an arbitrary function. If there exists a $(\delta(n),o(2^n \delta^2/n))$-compression algorithm for $\classc$, then $\nexp \nsubseteq \classc$. 
\end{proposition}

\begin{proof}
Let $\classc$ be a typical circuit class. Fix any function $\delta = \delta(n)$. Let $\mathcal{A}$ be an efficient $(\delta,o(2^n \delta^2/n))$-compression algorithm for $\classc$. We use $\mathcal{A}$ to construct an algorithm $\mathcal{B}$ that implicitly defines a property that is useful against $\classc$. The proof then follows immediately from Proposition \ref{p:usefullog}. 

We define $\mathcal{B}$ as follows. Given any truth table $tt(f) \in \{0,1\}^N$ as input, apply $\mathcal{A}$ to $tt(f)$ to obtain the description of a circuit $C$ over $n$ inputs. If $C$ is not a valid circuit, or it has more than $\alpha \cdot 2^n \delta^2/n$ gates, accept. Otherwise, check if $C$ computes $f$ with advantage $\delta$, and accepts $tt(f)$ if and only if this is \emph{not} the case.

Let $\propp$ be the property computed by $\mathcal{B}$. We need to check that $\propp$ is a $\mathsf{P}$-property that is useful against $\classc$. First, observe that $\mathcal{B}$ runs in time $\poly(N)$, since by assumption $\mathcal{A}$ is efficient, and $N = 2^n$. Also, $B$ will always accept some family of hard functions, since it follows from Lemma \ref{l:hardfunction} that for sufficiently large $n$ there are functions that cannot be computed with advantage $\delta$ by circuits of size less than $\alpha \cdot 2^n \delta^2 /n$. Finally, for any fixed $k$, it follows from the definition of lossy compression that there are infinitely many input sizes $n$ on which $\mathcal{A}$ succeeds. For all such inputs sizes, algorithm $\mathcal{B}$ will correctly reject functions computed by circuits from $\classc[n^k]$.
\end{proof}

This result is optimal for very small $\delta$. More precisely, it follows from elementary Fourier analysis of boolean functions that for every boolean function $f_n$ there is a parity function over some subset $S \subseteq [n]$ that computes $f_n$ with advantage $\Omega(2^{-n/2})$. Further, it is possible to check all parity functions in deterministic time $\poly(N)$.

\begin{remark}
Similar techniques can be used to show that lossy compression of quasipolynomial size circuits leads to circuit lower bounds for $\mathsf{NE} \cap \iocone$. This can be obtained through an application of Proposition \emph{\ref{p:unary}}.
\end{remark}

\subsection{Derandomization, SAT algorithms and circuit lower bounds}\label{s:derandomization}

In this section we use Williams' framework to prove that derandomization yields circuit lower bounds. Recall that $\mathsf{PIT}$ is the language consisting of all arithmetic circuits that compute the zero polynomial over $\mathbb{Z}$, and $\perm$ is the problem of computing the permanent of integer matrices.

Our proof uses the notion of useful algorithms introduced in Definition \ref{d:usefulalgorithms}. The following consequence is immediate from Proposition \ref{p:useful}.

\begin{corollary}\label{c:nexpppoly}
Assume that $\nexp \subseteq \size[\poly]$. Then there is $c\in \mathbb{N}$ such that there is no useful algorithm for $\equivand$-$\size[n^c]$.
\end{corollary}

In addition, we will need the following auxiliary lemma.

\begin{lemma}[Kabanets and Impagliazzo \citep{DBLP:journals/cc/KabanetsI04}, Aaronson and van Melkebeek \citep{DBLP:journals/toc/AaronsonM11}]\label{l:arith}~\\
There exists an efficient algorithm that takes as input an arithmetic circuit $A_m$ and
an integer $m$, and produces an arithmetic circuit $C_m$ such that $A_m$ computes the permanent of $m \times m$ matrices
matrices over $\mathbb{Z}$ if and only if $C_m \in \pit$.
\end{lemma}

We are now ready to give a short proof of the following result. Our argument follows the same high-level approach employed by \citep{DBLP:journals/cc/KabanetsI04} and \citep{DBLP:journals/toc/AaronsonM11}. 

\begin{proposition}[Kabanets and Impagliazzo \citep{DBLP:journals/cc/KabanetsI04}]\label{p:lowerfromderand}~\\
If $\pit \in \nsubexp$, then at least one of the following results hold:
\begin{itemize}
\item[\emph{(}i\emph{)}] $\nexp \nsubseteq \size[\poly(n)]$; or
\item[\emph{(}ii\emph{)}] $\perm \nsubseteq \asize[\poly(n)]$.
\end{itemize}
\end{proposition} 

\begin{proof}In order to derive a contradiction, assume that:
\begin{itemize}
\item $\pit \in \nsubexp$;
\item $\nexp \subseteq \size[\poly(n)]$;
\item $\perm \subseteq \asize[\poly(n)]$.
\end{itemize}

\noindent More precisely, $\nexp \subseteq \size[\poly(n)]$ implies that there exists a family of circuits $D = \{D_n\}_{n \in \mathbb{N}}$ of size $n^d$ that solves $\equivand$-$\size[n^c]$. In addition, $\perm$ over matrices of order $m$ can be solved by a family of arithmetic circuits $A = \{A_m\}_{m \in \mathbb{N}}$ of size $m^a$ (for some $a \in \mathbb{N}$). We prove that these assumptions contradict Corollary \ref{c:nexpppoly}. We construct a useful algorithm $\mathcal{A}$ for $\equivand$-$\size[n^c]$ as follows.~\\

\noindent {\bf Algorithm $\mathcal{A}$:}\\
\noindent {\bf Input:} Circuits $C_1,C_2$ of size $n^c$.
\begin{itemize}
\item First, $\mathcal{A}$ guesses a circuit $D_n$ of size $n^d$.
\item $\mathcal{A}$ prepares a query to the polynomial time hierarchy\footnote{Observe that $D_n$ does not solve the equivalence problem if and only if $ (\exists  C_1,C_2 \: \exists x$ such that $C_1(x) \neq C_2(x)$ and $D_n(C_1,C_2)=1)$ or $(\exists  C_1,C_2 $ such that $ \forall x (C_1(x) = C_2(x))$ and $D_n(C_1,C_2)=0)$.} to check if $D_n$ solves $\equivand$-$\size[n^c]$.
\item It uses Toda's theorem \citep{DBLP:journals/siamcomp/Toda91} together with the completeness of the permanent problem \citep{DBLP:journals/tcs/Valiant79} to reduce this query to a call to $\perm$ over matrices of dimension $m$, where $m = \poly(n^d)$.
\item Next, $\mathcal{A}$ guesses an arithmetic circuit $A_m$ of size $m^a$.
\item It then applies Lemma \ref{l:arith} to obtain a circuit $C_m$ such that $A_m$ computes the permanent of $m \times m$ matrices
matrices over $\mathbb{Z}$ if and only if $C_m \in \pit$.
\item Now $\mathcal{A}$ uses nondeterminism and the assumption that $\pit \in \nsubexp$ to check if $C_m \in \pit$. It aborts otherwise.
\item It uses $A_m$ to answer the initial query, and aborts if $D_n$ does not solve $\equivand$-$\size[n^c]$.
\item Finally, $\mathcal{A}$ uses $D_n$ to solve $\equivand$-$\size[n^c]$ on inputs $C_1$ and $C_2$.
\end{itemize}

Clearly, $\mathcal{A}$ runs in nondeterministic subexponential time. In addition, it is easy to see that it is a useful algorithm for $\equivand$-$\size[n^c]$, which completes the proof of Proposition \ref{p:clbfromderandomization}.
\end{proof}  

Most importantly, this proof shows that any improvement over Corollary \ref{c:nexpppoly} implies a corresponding improvement over Proposition \ref{p:lowerfromderand}. In addition, it is not hard to see that Conjecture \ref{c:unary} immediately implies the extension of Proposition \ref{p:lowerfromderand} obtained by Aaronson and van Melkebeek \citep{DBLP:journals/toc/AaronsonM11}\footnote{Here is a sketch of the argument. Assume that $\ne \cap \cone \subseteq \ppoly$. Then by Conjecture \ref{c:unary} there is no $\mathsf{P}$-property useful against $\ppoly$. However, it is possible to show that useful algorithms for satisfiability lead to useful properties. Altogether, these assumptions imply the desired strengthening of Corollary \ref{c:nexpppoly}.}.

\subsection{Useful properties and learning algorithms}\label{s:learning}

The existence of learning algorithms in many different models yields circuit lower bounds, as first shown by Fortnow and Klivans \citep{DBLP:journals/jcss/FortnowK09}. In this section we discuss two frameworks for learning: deterministic exact learning from membership and equivalence queries (Angluin \citep{Ang88}), and randomized PAC learning (Valiant \citep{DBLP:conf/stoc/Valiant84}).~\\

\noindent {\bf Exact learning algorithms.} Let $\classc$ be a typical circuit class. In this model, a \emph{deterministic} algorithm is given access to oracles $\mathsf{MQ}^f$ and $\mathsf{EQ}^f$ for some function $f:\Boolean$ in $\classc$. There oracles are defined as follows.~\\

\noindent $\mathsf{MQ}^f$: Given $x \in \{0,1\}^n$, returns $f(x)$.~\\

\noindent $\mathsf{EQ}^f$: Given a hypothesis $h:\Boolean$ represented as a circuit, returns $1$ if $h \equiv f$. Otherwise, returns an arbitrary input $x \in \{0,1\}^n$ such that $f(x) \neq h(x)$.~\\

\noindent For a size function $s: \mathbb{N} \rightarrow \mathbb{N}$, we say that a learning algorithm $\mathcal{A}$ \emph{exact learns} $\classc[s(n)]$ in time $t(n)$ if for every $f \in \classc$, when given access to oracles $\mathsf{MQ}^f$ and $\mathsf{EQ}^f$, $\mathcal{A}$ runs in time at most $t(n)$, and outputs the description of a circuit $C$ computing $f$. In particular, every equivalence query is invoked on a circuit of size at most $t(n)$, and the final hypothesis $C$ is a circuit of size at most $t(n)$.~\\

Recall that one of the main results from Fortnow and Klivans \citep{DBLP:journals/jcss/FortnowK09} states that exact learning a circuit class leads to circuit lower bounds against $\mathsf{E}^{\mathsf{NP}}$ (Proposition \ref{p:clbfromlearning}). The original proof used by them is a clever combination of many results from complexity theory. Here we observe that it is relatively easy to prove results of this form using the machinery of useful properties. To simplify the argument even more, we can view learning as compression, which yields a quick proof of the following result.

\begin{proposition}[``Learning yields circuit lower bounds'']\label{p:exactlearning}~\\
Let $\classc$ be a circuit class. Suppose there exists an exact learning algorithm for $\classc[\poly]$ that runs in subexponential time. Then $\nexp \nsubseteq \classc[\poly]$.
\end{proposition}

\begin{proof}
Let $\mathcal{A}$ be an exact learning algorithm for $\classc$. It is easy to see that given any truth-table $tt(h) \in \{0,1\}^N$ from $\classc[\poly]$, we can simulate $\mathcal{A}$ on input $h$ in time $2^{O(n)}$. In other words, it is possible to provide correct answers to the membership and equivalence queries asked during $\mathcal{A}$'s computation. By assumption, the learning algorithm outputs a circuit of subexponential size that computes $h$. This is therefore a valid compression algorithm for $\classc[\poly]$, and Proposition \ref{p:exactlearning} follows immediately from Proposition \ref{p:clbfromapprox} with $\delta = 1$. 
\end{proof}

In addition to its simplicity, this proof offers other advantages. The framework of useful properties is more flexible with respect to changes in the learning model. For instance, one could consider deterministic learning algorithms using equivalence queries over subsets $S \subseteq \{0,1\}^n$ encoded by subexponential size circuits, and only require that the learning algorithm outputs a hypothesis that is $\varepsilon$-close to the unknown concept. Again, Proposition \ref{p:clbfromapprox} easily implies circuit lower bounds.

Next we turn our attention to randomized learning algorithms, a class of algorithms for which theorems of the form ``learning implies circuit lower bounds'' are still much weaker than their deterministic counterpart.~\\

\noindent {\bf Randomized PAC learning algorithms.} In the PAC learning framework, there is an unknown function $f \in \classc$ that the learning algorithm is supposed to learn (after obtaining limited information about $f$). Here we concentrate on the stronger model in which the learner can ask membership queries, and only needs to learn under the uniform distribution\footnote{In other words, a transference theorem for this learning model is a stronger result. In addition, it is easy to see that the results discussed here hold under even more powerful learning models.}. In other words, the learner can query the value $f(x)$ on any input $x$, and should be able to obtain, with high probability, a good approximation $h$ for $f$. In general, for any function $f:\{0,1\}^n \rightarrow \{0,1\}$ in $\classc[s(n)]$, given parameters $n$, $\varepsilon$ (accuracy), $\delta$ (confidence), and an upper bound $s(n)$ on the size of the circuit computing $f$, the learning algorithm should output with probability at least $1 - \delta$ a hypothesis $h$  such that $\Pr_x [f(x) \neq h(x)] \leq \varepsilon$ (i.e., $h$ is $\varepsilon$-close to $f$), where the probability is taken over all strings $x$ of size $n$ under the uniform distribution. We measure the running time $t_\mathcal{A}(n,1/\delta, 1/\varepsilon, s(n))$ of a learning algorithm $\mathcal{A}$ as a function of these parameters. As opposed to what is usually called proper learning, the learning algorithm is allowed to output the description of any circuit of size at most $t_\mathcal{A}(.)$ as its final hypothesis. For simplicity, we say that an algorithm $\mathcal{A}$ PAC learns $\classc$ if it learns any function from $\classc$ to accuracy $1/4$ with probability at least $1 - 1/n$.~\\

It is known that the existence of a polynomial time PAC learning algorithm for $\classc[\poly]$ implies that $\mathsf{BPEXP} \nsubseteq \classc[\poly]$ (Fortnow and Klivans \citep{DBLP:journals/jcss/FortnowK09}). However, the same proof provides much weaker results for subexponential time learning, and it is an interesting open problem to show that the existence of subexponential time PAC learning algorithms lead to similar circuit lower bounds. The next proposition shows that this problem is related to the power of randomness in the context of useful properties. First, we extend the definition of useful properties to promise properties.

\begin{definition}[``Promise properties useful against $\classc$'']~\\
A \emph{promise property} of boolean functions $\propp = (\propp_{\mathsf{yes}}, \propp_{\mathsf{no}})$ consists of two nonempty disjoint subsets of the set of all boolean functions. For a typical circuit class $\classc$, $\propp$ is said to be \emph{useful} against $\classc$ if, for all $k$, there are infinitely many positive integers $n$ such that
\begin{itemize}
\item $\proppyes(f) = 1$ for at least one function $f:\Boolean$, and
\item $\proppno(g)=1$ for all $g:\Boolean$ that admits circuits from $\classc[n^k]$.
\end{itemize}
\noindent We say that a promise property $\propp$ is a $\Gamma$-\emph{property} if its corresponding promise problem $L_{\propp}$ is in $\mathsf{promise}$-$\Gamma$.
\end{definition}

\begin{proposition}[``Useful properties from randomized learning'']\label{p:paclearning}~\\
Let $\classc$ be a typical circuit class. Suppose there exists a randomized algorithm $\mathcal{A}$ that \emph{PAC} learns $\classc[\poly]$ in time $2^{n^{o(1)}}$. Then there exists a $(\mathsf{promise}$-$\mathsf{coRP})$-property that is useful against $\classc$.
\end{proposition}

\begin{proof}
We use a subexponential time randomized learning algorithm $\mathcal{A}$ for $\classc$ to define a (promise) $\mathsf{coRP}$-property $\propp$ that is useful against $\classc$. Consider the following randomized algorithm $\mathcal{B}$. Given the truth-table $tt(f_n) \in \{0,1\}^N$ of an arbitrary function $f_n:\Boolean$, it simulates the computation of $\mathcal{A}$ over $f_n$, until $\mathcal{A}$ outputs a circuit $C$ of size $2^{n^{o(1)}}$ as its final hypothesis. Algorithm $\mathcal{B}$ accepts $f_n$ if and only if $C$ is \emph{not} $1/10$-close to $f_n$. 

It follows from Lemma \ref{l:hardfunction} that for any large enough $n$ there is a function $h_n$ that cannot be $1/10$-approximated by circuits of subexponential size (for definiteness, fix some constructive size bound). In other words, for any large $n$, there exists at least one function $h_n$ not in $\classc[\poly]$ that is accepted with probability one. In addition, since $\mathcal{A}$ is a PAC learning algorithm for $\classc$, every function in $\classc$ is rejected with high probability. Clearly, $\mathcal{B}$ computes a promise $\mathsf{coRP}$-property that is useful against $\classc$: $\proppyes$ consists of boolean functions that cannot be approximated by circuits of subexponential size, and $\proppno = \classc$. 
\end{proof}

This result gives another example of the fundamental importance of the notion of useful properties in the context of results of the form ``algorithms yield circuit lower bounds''.

\section{Some broad research directions} \label{s:conclusion}

Here is a list of problems related to the results discussed in this survey that we find particularly interesting.~\\

\noindent {\bf Strengthening the $\acc$ lower bound.} Williams proved that $\nexp \nsubseteq \acc$. It follows easily from Lemma \ref{l:trick} that either $\mathsf{P} \nsubseteq \mathsf{ACC}$ or $\mathsf{NEXP} \nsubseteq \mathsf{P}/\mathsf{poly}$. Give an unconditional proof that one of these circuit lower bounds hold.~\\

\noindent {\bf Lossy compression of $\acc$ and $\tc_2$.} Design efficient lossy compression schemes for circuit classes such as $\acc$ or $\tc_2$. To the best of our knowledge, these results do not violate any widely believed cryptographic assumption.~\\

\noindent {\bf Satisfiability algorithms.} Can we make progress on satisfiability algorithms for threshold circuits?~\\

\noindent {\bf SAT algorithms for depth $d+1$ versus lower bounds against depth $d$.} Is it possible to show that, in general, nontrivial satisfiability algorithms for $\classc_{d+1}$ lead to circuit lower bounds against $\classc_d$?

\section*{Acknowledgements}

I would like to thank Rocco Servedio for helpful conversations that lead to the simplification of some proofs. I would also like to thank Valentine Kabanets and Ryan Williams for reading and commenting on a first draft of this survey. Finally, Cl\'{e}ment Canonne provided me valuable advice to improve the presentation.
  
\bibliographystyle{alpha}	
\bibliography{refs-full}	

\appendix

\newpage

\section{$\ne \cap \iocone$ lower bounds from useful properties} \label{a:unary}

In this section we describe the proof of Proposition \ref{p:unary}, which we state again for convenience. 

\begin{proposition*}
Let $\classc$ be a typical circuit class. If for every $c \in \mathbb{N}$ there exists a $\mathsf{P}$-property that is useful against $\classc[n^{\log^{c}n}]$, then $\ne \cap \iocone \nsubseteq \classc[n^{\log n}]$.
\end{proposition*}

This result is implicit in the work of Williams \citep{DBLP:conf/stoc/Williams13}, and it consists of an interesting combination of nondeterminism, a collapse theorem, a hardness vs. randomness result, and simple diagonalization. We will need the following auxiliary results.

\begin{lemma}\label{l:ctosize}
Let $\classc$ be a typical circuit class, and assume that $\mathsf{P} \subseteq \classc[n^{\log n}]$. Then for every $d \in \mathbb{N}$, any function $f:\{0,1\}^n \rightarrow \{0,1\}$ computed by circuits of size $n^{\log^{d}n}$ is computed by circuits from $\classc[n^{\log^{O(d)}n}]$.
\end{lemma}

\begin{proof}
The result follows from a parameterized version of Lemma \ref{l:trick}, and the proof is similar. 
\end{proof}

\begin{lemma}[Miltersen, Vinodchandran and Watanabe \citep{DBLP:conf/cocoon/MiltersenVW99}]\label{l:mvw}~\\
Let $g(n) > 2^n$ and $s(n) \geq n$ be functions that are both increasing and time-constructible. There exists a constant $d \in \mathbb{N}$ for which the following holds. If $\mathsf{E} \subseteq \mathsf{SIZE}(s(n))$ then $\mathsf{DTIME}[g(n)] \subseteq \mathsf{MATIME}[s(d \log g(n))^d]$.
\end{lemma}

For a function $h_\ell:\{0,1\}^\ell \rightarrow \{0,1\}$, let $\mathsf{CC}(h)$ be the size (number of gates) of the smallest circuit computing $h$.

\begin{proposition}[Umans \citep{DBLP:journals/jcss/Umans03}]\label{p:umans}~\\
There is a constant $k \in \mathbb{N}$ and a function $G: \{0,1\}^* \times \{0,1\}^* \rightarrow \{0,1\}^*$ for which the following holds. For every $s \in \mathbb{N}$ and boolean function $h_\ell:\{0,1\}^\ell \rightarrow \{0,1\}$ satisfying $\mathsf{CC}(h_\ell) \geq s^k$, and for all circuits $C$ of size at most $s$ over $s$ inputs, 
$$
\left | \Pr_{z \in \{0,1\}^{k \cdot \ell}} \left [ C(G(tt(h_\ell),z)) = 1 \right ] - \Pr_{z \in \{0,1\}^s} \left [ C(z) = 1 \right ]  \right | < \frac{1}{s}.
$$
In addition, $G$ can be computed in $\mathsf{poly}(2^\ell)$ time.
\end{proposition}

The next lemma shows that useful properties together with the lack of circuit lower bounds for $\mathsf{P}$ allow us to obtain a nontrivial derandomization of Merlin-Arthur games.

\begin{lemma}\label{l:ma}
Let $\classc$ be a typical circuit class, and suppose that for every $c \in \mathbb{N}$ there exists a $\mathsf{P}$-property that is useful against $\classc[n^{\log^{c}n}]$. In addition, assume that $\mathsf{P} \subseteq \classc[n^{\log n}]$. Then there is an infinite subset $S \subseteq \mathbb{N}$ such that for any $L \in \matime[n^{O(\log^3 n)}]$, there exists a language $L' \in \ne $ such that for every $n \in S$, we have $L_n = L'_n$. In addition, for all $n \notin S$, we have $L'_n = \emptyset$.
\end{lemma}

\begin{proof}
First, observe that Lemma \ref{l:ctosize} implies that for every $c \in \mathbb{N}$ there exists a property $\propp_c$ that is useful against $\mathsf{SIZE}[n^{\log^c n}]$. Let $\mathcal{A}^c$ be an efficient algorithm computing $\propp_c$ (we set the value of $c$ later). Let $L \in \matime[n^{O(\log^3 n)}]$. There exists a $\mathsf{MA}$-verifier $V$ for $L$ running in time $s = n^{O(\log^3 n)}$ such that
\begin{eqnarray}
x \in L \quad  & \Longrightarrow & \quad \exists y \in \{0,1\}^s \Pr_{w \in \{0,1\}^s}\left[V(x,y,w)= 1 \right] \geq \frac{2}{3} \nonumber \\ 
x \notin L \quad  & \Longrightarrow & \quad \forall y \in \{0,1\}^s \Pr_{w \in \{0,1\}^s}\left[V(x,y,w)= 1 \right] \leq \frac{1}{3} \nonumber
\end{eqnarray}

Our nondeterministic algorithm $\mathcal{N}$ for $L$ proceeds as follows. On input $x \in \{0,1\}^n$, it first guesses a string $y \in \{0,1\}^s$, then constructs a circuit $C_{x,y}$ from $\mathsf{SIZE}[s]$ such that for all $w \in \{0,1\}^s$ we have $C_{x,y}(w) = V(x,y,w)$. Then $\mathcal{N}$ guesses truth-tables $tt(h_m) \in \{0,1\}^M$ for every $m \in [2^{(\log n)^{5/(c+1)}}, 2^{(\log (n+1))^{5/(c+1)}})$, where $M = 2^m$ as usual. If $\mathcal{A}^c$ rejects all such functions, then $\mathcal{N}$ rejects $x$. Otherwise, let $h_\ell$ be the first function for which $\mathcal{A}^c(h_\ell) = 1$. Since $\mathcal{A}^c$ computes a useful property, $s = n^{O(\log^3 n)}$, and $\ell \geq 2^{(\log n)^{5/(c+1)}}$, for any $c \in \mathbb{N}$ we have:
$$
\mathsf{CC}(h_\ell) \geq \ell^{\log^c \ell} \geq n^{\log^4 n} \gg s^k,
$$
for any $k \in \mathbb{N}$ and sufficiently large $n$. Finally, $\mathcal{N}$ runs the algorithm granted by Proposition \ref{p:umans} on $C_{x,y}$ using $h_\ell$, and accepts its input $x$ if and only if
\begin{equation}\label{eq:prob}
\Pr_{z \in \{0,1\}^{k \cdot \ell}}\left[ C_{x,y}(G(tt(h_\ell),z))= 1   \right] \geq \frac{1}{2}.
\end{equation}

Observe that there exists an infinite set $S \subseteq \mathbb{N}$ such that for each $n \in S$ and for every $x \in \{0,1\}^n$, $\mathcal{N}$ is able to find a function $h_\ell$ for which $\mathsf{CC}(h_\ell) \geq s^k$, where $k$ is the constant in the statement of Proposition \ref{p:umans}. Put another way, $\mathcal{N}$ is correct on input sizes in $S$, and by construction $\mathcal{N}$ rejects every other input whose input size is not in $S$. Also, $S$ depends only on $\propp_c$.

The (nondeterministic) running time of $\mathcal{N}$ is dominated by the computation of the probability in (\ref{eq:prob}), and the time required to verify using $\mathcal{A}^c$ whether some hard function has been guessed. Finally, set $c = 5$, and observe that for this value of $c$ we have $\ell \ll n$. It follows therefore that $\mathcal{N}$ runs in time at most $2^n$. This completes the proof that there exists $L' \in \ne$ such that for every $n \in S$, $L'_n = L_n$, and for all $n \notin S$, we have $L'_n = \emptyset$.
\end{proof}

We are now ready to give the proof of Proposition \ref{p:unary}.

\begin{proof}[Proof of Proposition \ref{p:unary}]
Assume that $\ne \cap \iocone \subseteq \classc[n^{\log n}]$. In particular, $\mathsf{E} \subseteq \mathsf{SIZE}[n^{\log n}]$. Let $g(n) = 2^{n^{2 \log n}}$ and $s(n) = n^{\log n}$. Using Lemma \ref{l:mvw}, we get $\dtime[2^{n^{2\log n}}] \subseteq \matime[n^{O(\log^3 n)}]$. Clearly, our assumptions also imply that $\mathsf{P} \subseteq \classc[n^{\log n}]$. 

Let $L \in \dtime[2^{n^{2\log n}}]$. It follows from Lemma \ref{l:ma} that there exists an infinite set $S \subseteq \mathbb{N}$ and a language $L' \in \ne$ such that $L_n = L'_n$ for every $n \in S$. Consider $\overline{L} \in \dtime[2^{n^{2\log n}}]$, the complement of $L$. Then, again, there exists a language $L'' \in \ne$ such that for every $n \in S$, $\overline{L}_n = L''_n$. Clearly, $\overline{L''} \in \cone$, and for every $n \in S$ we have $\overline{L''}_n= L_n = L'_n$. In other words, $L' \in \ne \cap \iocone$. Overall, we get
$$
\dtime[2^{n^{2\log n}}] \subseteq \mathsf{i.o.}(\ne \cap \iocone) \subseteq \mathsf{i.o.}\classc[n^{\log n}],
$$
where the last inclusion uses our initial assumption.

However, using a simple diagonalization argument, we can define a language $L^* \in \dtime[2^{n^{2\log n}}]$ such that for all $n \geq n_0$, $L^{*}_n$ is not computed by circuits from $\classc[n^{\log n}]$. This contradiction completes the proof of Proposition \ref{p:unary}.
\end{proof}
 
\end{document}